\documentclass[11pt]{article}

\usepackage{fullpage}  
\usepackage{float}
\usepackage{pb-diagram}  		

\usepackage{url}

\usepackage{algorithm} 
\usepackage{algorithmic} 
\usepackage{multirow} 
\usepackage{hhline}  
\usepackage{amsmath}
\usepackage{xcolor}

\usepackage{graphicx}
\usepackage{amscd}
\usepackage{amssymb,mathrsfs}
\input{epsf.sty}
\usepackage{amsthm,amscd}
\usepackage{color}
\usepackage{latexsym}
\usepackage{epic}
\usepackage{appendix}
\usepackage{enumerate}
\usepackage{longtable}
\usepackage{lscape}
\usepackage{extarrows}
\usepackage{epstopdf}
\usepackage[driverfallback=dvipdfm]{hyperref}

\usepackage{verbatim}

\newlength\myindent

\newcommand{\whcomm}[2]{{}{#2}}
\newcommand{\whfirrev}[2]{{}{#2}}

\newcommand{\comm}[1]{}

\newcommand{\pacomm}[1]{{}}

\newcommand{\delete}[1]{}

\newcommand{\Rmnum}[1]{\expandafter\@slowromancap\romannumeral #1@}
\newcommand{\inner}[3][]{{\left\langle #2,#3 \right\rangle_{#1}}}

\newtheorem{theo}{Theorem}[section]
\newtheorem{lemm}{Lemma}[section]

\newtheorem{rema}{Remark}[section]

\newtheorem{cond}{Condition}[section]

\numberwithin{equation}{section}

\DeclareMathOperator{\T}{\mathrm{T}}
\DeclareMathOperator{\Hess}{\mathrm{Hess}}
\DeclareMathOperator{\grad}{\mathrm{grad}}

\DeclareMathOperator{\D}{\mathrm{D}}

\DeclareMathOperator{\trace}{\mathrm{trace}}

\DeclareMathOperator{\diag}{\mathrm{diag}}

\DeclareMathOperator{\GL}{\mathrm{GL}}

\DeclareMathOperator{\rank}{\mathrm{rank}}

\begin{document}

\title{Blind Deconvolution by a Steepest Descent Algorithm on a Quotient Manifold}
\author{
Wen Huang 
    \thanks{
            huwst08@gmail.com,
            Department of Computational and Applied Mathematics,
            Rice University,
            Houston, TX 77005-1827, USA
    } 
              \and 
Paul Hand
    \thanks{Department of Computational and Applied Mathematics,
            Rice University,
            Houston, TX 77005-1827, USA
    } 
}

\maketitle


\begin{abstract}
In this paper, we propose a Riemannian steepest descent method for solving a blind deconvolution problem. We prove that the proposed algorithm with an appropriate initialization will recover the exact solution with high probability when the number of measurements is, up to log-factors, the information-theoretical minimum scaling. The quotient structure in our formulation yields a simpler penalty term in the cost function compared to~\cite{LLSW2016}, which eases the convergence analysis and yields a natural implementation. Empirically, the proposed algorithm has better performance than the Wirtinger gradient descent algorithm and an alternating minimization algorithm in the sense that i) it needs fewer operations, such as DFTs and matrix-vector multiplications, to reach a similar accuracy, and ii) it has a higher probability of successful recovery in synthetic tests. \whcomm{}{An image deblurring problem is also used to demonstrate the efficiency and effectiveness of the proposed algorithm.}
\end{abstract}

\section{Introduction}

We consider the problem of separating two unknown signals $\mathbf{w} \in \mathbb{C}^L$ and $\mathbf{x} \in \mathbb{C}^L$ given their circular convolution $\mathbf{y} \in \mathbb{C}^L$, which is an instance of blind deconvolution. Blind deconvolution is of interest due to many applications, e.g., astronomy, medical imaging, optics and communications engineering~\cite{JC1993,WP1998,CamEgi2007,LWDF2011,WBSJ2015}. The blind deconvolution problem is ill-posed without any additional assumptions. One commonly-used assumption is to suppose that the two signals $\mathbf{w}$ and $\mathbf{x}$ belong to known subspaces~\cite{ARR2014,LLJB2015,LingStro2015}. Specifically, the signal $\mathbf{w}$ and $\mathbf{x}$ can be written as
\begin{equation*} \label{BD:e6}
\mathbf{w} = \mathbf{B} h, h \in \mathbb{C}^K \hbox{ and } \mathbf{x} = \overline{\mathbf{C} m}, m \in \mathbb{C}^N,
\end{equation*}
for some matrices $\mathbf{B} \in \mathbb{C}^{L \times K}$ and $\mathbf{C} \in \mathbb{C}^{L \times N}$, where overbar denotes the complex conjugate.
Therefore, the blind deconvolution model with noise $\mathbf{e} \in \mathbb{C}^L$ is to find $h$ amd $m$ such that
\begin{equation} \label{BD:e7}
\mathbf{y} = \mathbf{w} * \mathbf{x} + \whfirrev{}{\mathbf{e}} = \mathbf{B} h * \overline{\mathbf{C} m} + \mathbf{e},
\end{equation}
given $\mathbf{y} \in \mathbb{C}^L, \mathbf{B} \in \mathbb{C}^{L \times K}$ and $\mathbf{C} \in \mathbb{C}^{L \times N}$, where $*$ denotes circular convolution, i.e., $\mathbf{y}_i = \sum_{j = 1}^{L} \mathbf{w}_j \mathbf{x}_{i + 1 - j}$ and the subscript $i + 1 - j$ of $\mathbf{x}$ is understood to be modulo $\{1, 2, \ldots, L\}$. For theoretical and numerical purposes, we express~\eqref{BD:e7} in the Fourier domain, see e.g.,~\cite{ARR2014,LLSW2016}. Let $\mathbf{F}$ denote the $L \times L$ unitary Discrete Fourier Transform (DFT) matrix.
Taking Fourier transform for both sides of~\eqref{BD:e7} yields
\begin{equation} \label{BD:e8}
\frac{\mathbf{F} \mathbf{y}}{\sqrt{L} } = (\mathbf{F} \mathbf{B} h) \odot (\mathbf{F} \overline{\mathbf{C} m}) + \frac{\mathbf{F} \mathbf{e}}{\sqrt{L}},
\end{equation}
where $\odot$ denotes the Hadamard product. Throughout this paper, we denote $\mathbf{F} \mathbf{y} / \sqrt{L}$, $\mathbf{F} \mathbf{B}$, $\mathbf{F} \bar{\mathbf{C}}$, and $\mathbf{F} \mathbf{e} / \sqrt{L}$ by $y$, $B$, $\bar{C}$, and $e$ respectively. Therefore, \eqref{BD:e8} becomes $y = B h \odot \overline{C m} + e$ and the blind deconvolution problem is equivalent to
\begin{equation} \label{BD:e9}
\hbox{find $h \in \mathbb{C}^K, m \in \mathbb{C}^N$ such that } y = B h \odot \overline{C m} + e.
\end{equation}

In recent years, two important frameworks~\cite{ARR2014,LLSW2016} have been proposed to solve~\eqref{BD:e9} and
to admit a rigorous recovery guarantee. In~\cite{ARR2014}, a convex optimization framework is proposed for the blind deconvolution problem. Specifically, the problem~\eqref{BD:e9} is reformulated into a problem of recovering a rank-one matrix from an underdetermined system of linear equations by the well-known lifting trick\footnote{In~\cite{ARR2014}, the vectors $h$ and $m$ are in real domain, i.e., $h \in \mathbb{R}^K$ and $m \in \mathbb{R}^N$.}.
Replacing the rank-one constraint with a nuclear norm penalty converts the original problem into a semidefinite program.
Theoretical results given in the paper show that
this program enjoys a recovery guarantee under reasonable condition, such as appropriate sample complexity. The primary drawback is that the computational cost of the semidefinite program is high for large-scale problems since the lifting trick squares the dimension of the problem.
The paper indeed proposes an algorithm which is based on the matrix-factorization in~\cite{BurMon03}; however, the convergence guarantee to the solution is not provided therein.
The paper~\cite{LLSW2016} overcomes the difficulty by resorting to a nonconvex approach directly. The idea closely follows from~\cite{CLS2016}, which shows that with high probability i) a point sufficiently close to the solution can be computed, and ii) the proposed algorithm converges to the desired solution with the computed point as the initial iterate.
The cost function in~\cite{LLSW2016} is
\begin{equation} \label{BD:e10}
F(h, m) = \|y - B h \odot \overline{C m}\|_2^2 = \|y - \diag(B h m^* C^*)\|_2^2.
\end{equation}
The minimizer of $F(h, m)$ is not unique since a solution $(h_{\sharp}, m_{\sharp})$ implies that $(h_{\sharp} p^{-1}, m_{\sharp} p)$, for any $p \in \mathbb{C}$, $p\neq 0$ is also a solution. To ensure numerical stability, the authors in~\cite{LLSW2016} add a penalty function to~\eqref{BD:e10} so that $\|h_{\sharp}\|_2$ and $\|m_{\sharp}\|_2$ are roughly equal in magnitude.\footnote{The penalty term in~\cite{LLSW2016} also includes a term for coherence between $h$ and the rows of $B$. This will be discussed later.} Note that $h m^*, h \neq 0, m \neq 0$ is a rank-one matrix. The cost function~\eqref{BD:e10} is well-defined as a function:
\begin{align*}
f:&\mathbb{C}_1^{K \times N} \rightarrow \mathbb{R}\\
&X \mapsto \|y - \diag(B X C^*)\|_2^2,
\end{align*}
where $\mathbb{C}_1^{K \times N} = \{X \in \mathbb{C}^{K \times N} \mid \rank(X) = 1\}$ denotes the set of rank-one matrices.
In practice, working $\mathbb{C}^{K \times N}$ is usually inefficient especially when $K$ and $N$ are large.
Therefore, we resort to the quotient manifold $\mathbb{C}_*^{K} \times \mathbb{C}_*^{N} / \mathbb{C}_*$ since this manifold is a diffeomorphism to $\mathbb{C}_1^{K \times N}$, where $\mathbb{C}_*^n$ denotes a non-zero $n$-dimensional vector (the subscript $n$ is dropped when $n = 1$), and the definition of $\mathbb{C}_*^{K} \times \mathbb{C}_*^{N} / \mathbb{C}_*$ is given in Sections~\ref{BD:s11} and~\ref{BD:s4}. 
A representation of a point in $\mathbb{C}_*^{K} \times \mathbb{C}_*^{N} / \mathbb{C}_*$ is in the space $\mathbb{C}_*^{K} \times \mathbb{C}_*^{N}$, which requires much smaller storage than $\mathbb{C}^{K \times N}$. By slightly abusing notation, the cost function defined on the quotient manifold is
\begin{align} \label{BD:e14}
f:&\mathbb{C}_*^{K} \times \mathbb{C}_*^{N} / \mathbb{C}_* \rightarrow \mathbb{R} \nonumber \\
&\pi(h, m) \mapsto \|y - \diag(B h m^* C^*)\|_2^2,
\end{align}
where $\pi(h, m)$ denotes the point in the quotient manifold corresponding to $(h, m)$. We refer to Sections~\ref{BD:s11} and~\ref{BD:s4} for the more detailed discussions about the quotient manifold.
One can apply manifold optimization techniques to find the minimizer of~\eqref{BD:e14} directly.

In this paper, we propose a Riemannian steepest descent method for solving the blind deconvolution problem~\eqref{BD:e7}. We prove that the proposed algorithm with an appropriate initialization will recover the exact solution with high probability when the number of measurements is, up to log-factors, the information-theoretical minimum scaling. The quotient structure in our formulation yields a simpler penalty term in the cost function compared to~\cite{LLSW2016}, which eases the convergence analysis and yields a natural implementation. In addition, unlike the cost function in~\cite{LLSW2016}, the objective $f$ is strongly convex in a neighborhood of the solution, which yields efficiency in optimization algorithms since the convergence rate of a steepest descent method is related to the condition number of the Hessian at the minimizer~\cite{NocWri2006,AMS2008}.
Empirically, the proposed algorithm has better performance than the Wirtinger gradient descent algorithm in~\cite{LLSW2016} and an alternating minimization algorithm, see e.g.,~\cite{You1994} in the sense that i) it needs fewer operations, such as DFTs and matrix-vector multiplications, to reach a similar accuracy, and ii) it has a higher probability of successful recovery.


Applying Riemannian optimization methods for solving signal processing problems, of course, is not new. To name a few, a Riemannian trust-region Newton method is used in~\cite{Boumal2013} to solve the synchronization of rotations problem. Various Riemannian methods are used to solve the matrix completion problem~\cite{BouAbs2011,Vandereycken2013,Mishra2014,HAGH2016,WCCL2016}. A limited memory version of Riemannian BFGS method with adaptive rank strategy is used to solve the phase retrieval problem~\cite{HGZ2017}. It has been shown that Riemannian methods perform very well and are competitive, and sometimes the best, methods in many applications.

In this paper, we will use the following notation.
$\mathbb{C}_r^{K \times N} = \{X \in \mathbb{C}^{K \times N} \mid \rank(X) = r\}$ denotes the set of fixed rank matrices. 
$\mathbb{C}_*^{n \times p} = \{X \in \mathbb{C}^{n \times p} \mid \rank(X) = p\}$ denotes the noncompact Stiefel manifold. In particular, $\mathbb{C}_*^n$ denotes the set of non-zero vectors. $\GL(r) = \{X \in \mathbb{C}^{r \times r} \mid \rank(X) = r\}$ denotes the generalized linear group.
Superscript $*$ denotes the conjugate transpose operator. Overbar $\bar{\cdot}$ denotes the complex conjugate. $B$ is denoted by $[b_1 b_2 \ldots b_L]^*$ and $[\mathfrak{b}_1 \mathfrak{b}_2 \ldots \mathfrak{b}_K]$; $C$ denotes $[c_1 c_2 \ldots c_L]^*$ and $[\mathfrak{c}_1 \mathfrak{c}_2 \ldots \mathfrak{c}_N]$;
$\mathcal{A}$ denote the linear operator $\mathcal{A}(Z) = \{b_l^* Z c_l\}_{l=1}^L = \diag(B Z C^*)$. The adjoint operator $\mathcal{A}^*$ therefore is $\mathcal{A}^*(z) = \sum_{l = 1}^L z_l b_l c_l^* = B^* \diag(z) C$.
Given a function $\mathbf{f}: \mathcal{M} \rightarrow \mathbb{R}$ where $\mathcal{M} \subset \mathbb{C}^{n}$ is a Riemannian manifold, $\nabla \mathbf{f}$ denotes the Euclidean gradient of $\mathbf{f}$, i.e., with respect to the metric $\inner[2]{\eta}{\xi} = \mathrm{Re}(\trace(\eta^* \xi))$ and $\grad \mathbf{f}$ denotes the Riemannian gradient of $\mathbf{f}$, i.e., with respect to the Riemannian metric of $\mathcal{M}$. Given a function $\mathbf{h}: \mathcal{M}_1 \rightarrow \mathcal{M}_2$, notation $\D \mathbf{h}(x) [\eta_x]$ denotes the directional derivative of $h$ at $x$ along direction $\eta_x \in \T_x \mathcal{M}_1$, where $\mathcal{M}_1$ and $\mathcal{M}_2$ are two manifolds, and $\T_x \mathcal{M}_1$ is the tangent space of $\mathcal{M}_1$ at $x$.

This paper is organized as follows. Section~\ref{BD:s11} presents the Riemannian framework. Section~\ref{BD:s14} gives the initialization, proposed Riemannian method and its convergent results. Section~\ref{BD:s12} derives the Riemannian ingredients of the quotient manifold. Section~\ref{BD:s15} reports numerical experiments and Section~\ref{BD:s16} states the conclusion.

\section{Problem Statement} \label{BD:s11}

In order to state the method, we first introduce coherence.
As observed in~\cite{ARR2014,LLSW2016}, the number of measurements required for solving the blind deconvolution problem depends on how much $h_{\sharp}$ is correlated with the rows of the matrix $B$. This phenomenon also holds for the Riemannian framework. The coherence between the rows of $B$ and $h_{\sharp}$ is defined as\footnote{In~\cite[(3.2)]{LLSW2016}, the authors use the term 'incoherence' rather than 'coherence'.}
\begin{equation*}
\mu_h^2 = \frac{L \|B h_{\sharp}\|_{\infty}^2}{\|h_{\sharp}\|_2^2}.
\end{equation*}
Note that a small coherence $\mu_h$ is preferred.


For a fixed-rank manifold in real space, multiple representations have been proposed to represent a point in $\mathbb{R}_r^{K \times N}$, see~\cite{Vandereycken2013,Mishra2014}. The same idea can be applied for the fixed-rank manifold in complex space.
In this paper, the quotient manifold $\mathbb{C}_*^{K} \times \mathbb{C}_*^{N} / \GL(1) = \mathbb{C}_*^{K} \times \mathbb{C}_*^{N} / \mathbb{C}_*$ is used, where $\mathbb{C}_*^{K} \times \mathbb{C}_*^{N} / \mathbb{C}_* = \{[(h, m)] \mid (h, m) \in \mathbb{C}_*^{K} \times \mathbb{C}_*^{N}\}$, and $[(h, m)] = \{(h p^{-1}, m p^*) \mid p \in \mathbb{C}_*\}$ is the orbit of an element $(h, m) \in \mathbb{C}_*^{K} \times \mathbb{C}_*^{N}$ under the the group $\mathbb{C}_*$. The group action of $\mathbb{C}_*$ on $\mathbb{C}_*^{K} \times \mathbb{C}_*^{N}$ is $(h, m) \bullet p = (h p^{-1}, m p^*)$, where $p \in \mathbb{C}_*$ and $(h, m) \in \mathbb{C}_*^{K} \times \mathbb{C}_*^{N}$. Note that the manifolds here are over $\mathbb{R}$, i.e., it has a real parameterization.

The cost function that we consider in the Riemannian framework is
\begin{align}
\tilde{f}:&\mathbb{C}_*^{K} \times \mathbb{C}_*^{N} / \mathbb{C}_* \rightarrow \mathbb{R} \nonumber \\
&\pi(h, m) \mapsto \|y - \diag(B h m^* C^*)\|_2^2 + G(\pi(h, m)), \label{BD:e13} 
\end{align}
where
\begin{equation} \label{BD:e36}
G(\pi(h, m)) = \rho \sum_{i = 1}^L G_0\left(\frac{L |b_i^* h|^2 \|m\|_2^2}{8 d^2 \mu^2}\right),
\end{equation}
$G_0(t) = \max(t - 1, 0)^2$ is a $C^1$ function, $\rho \geq d^2 + 2.5 \|e\|_2^2$, $\frac{9}{10} d_* \leq d \leq \frac{11}{10} d_*$, $\mu \geq \mu_h:= \sqrt{L} \|B h_{\sharp}\|_\infty / \|h_{\sharp}\|_{\whfirrev{}{2}}$ and $d_* = \|h_{\sharp} m_{\sharp}^T\|_F = \|h_{\sharp}\|_2 \|m_{\sharp}\|_2$. The penalty term $G$ promotes a small coherence. The parameter $d$ can be computed by a spectral initialization (see Algorithm~\ref{BD:a4}). The value of $\rho$ can be selected by $\rho \geq d^2 + 2.5 \|e\|_2^2$ and in the case that $e$ is Gaussian, we choose $\rho \approx d^2 + 2.5 \sigma^2 d_*^2$ where $\|e\|_2^2 \sim \frac{\sigma^2 d_*^2}{2L} \chi_{2L}^2$, $\chi_{2L}^2$ is the $\chi$-squared distribution and therefore $\|e\|_2^2$ concentrates around $\sigma^2 d_*^2$. The parameter \whfirrev{}{$\mu$} can be selected as described \whfirrev{}{in the last paragraph of}~\cite[Section~3.2]{LLSW2016}.

The framework~\eqref{BD:e13}-~\eqref{BD:e36} is a simplified version of that in~\cite{LLSW2016}, where the cost function is
\begin{align}
\tilde{F}:&\mathbb{C}^K \times \mathbb{C}^N \rightarrow \mathbb{R} \nonumber \\
& (h, m) \mapsto \tilde{F}(h, x) = \|y - \diag(B h m^* C^*)\|_2^2 + \tilde{G}(h, m), \label{BD:e12}
\end{align}
where
\begin{equation} \label{BD:e35}
\tilde{G}(h, m) = \rho \left[G_0\left(\frac{\|h\|_2^2}{2 d}\right) + G_0\left(\frac{\|m\|_2^2}{2 d}\right) + \sum_{i = 1}^L G_0\left(\frac{L |b_i^* h|^2}{8 d \mu^2}\right)\right]
\end{equation}
is the penalty term, $G_0$, $\rho$, $d$, $\mu$ are defined as those in~\eqref{BD:e36}. The first two terms in $\tilde{G}$ penalize large values of $\|h\|_2$ and $\|m\|_2$. It follows that the norms of $h$ and $m$ are not different too much, which avoids a numerical stability problem. \whfirrev{}{The third term in $\tilde{G}$ comes from an earlier work~\cite{SL2016} for matrix completion.} Note that a discussion about the choices of $G_0$, $\rho$, $d$ and $\mu$ has been given in~\cite[Sections~3.1 and~3.2]{LLSW2016}.

The framework in~\cite{LLSW2016} and the Riemannian framework are similar in the sense that their cost functions, ~\eqref{BD:e13} and~\eqref{BD:e12}, are non-convex. Here, we emphasize their differences:
\begin{itemize}
\item the penalty term in~\eqref{BD:e13} scales by $\|m\|_2^2$ so that it is well-defined on  the quotient manifold.
\item the penalty term in~\eqref{BD:e13} does not directly penalize $\|m\|_2$ and $\|h\|_2$ because the iterates $\{\pi(h_k, m_k)\}$ generated by an optimization algorithm is invariant to the representation in orbits $\{[h_k, m_k]\}$. Therefore, one can choose any representation, such as $(h_k, m_k)$ satisfying $\|h_k\|_2 = \|m_k\|_2$.
\end{itemize}

\section{Optimization Algorithm and Theoretical Results} \label{BD:s14}


In this section, we state the proposed optimization algorithm for minimizing the function~\eqref{BD:e13} as well as the main theoretical results. Because the cost function~\eqref{BD:e13} is nonconvex, the proposed algorithm has two phases: initialization and Riemannian steepest descent method stated in Sections~\ref{BD:s21} and~\ref{BD:s5} respectively.

We consider the following model:
\begin{itemize}
\item the noise $e$ is drawn from the Gaussian distribution $\mathbb{N}(0, \frac{\sigma^2 d_*^2}{2 L} I_L) + \sqrt{-1} \mathbb{N}(0, \frac{\sigma^2 d_*^2}{2 L} I_L)$;
\item $C$ is a complex Gaussian distribution, i.e., $C_{i j} \sim \mathbb{N}(0, \frac{1}{2}) + \sqrt{-1} \mathbb{N}(0, \frac{1}{2})$; and
\item $B$ satisfies $B^* B = I_K$ and $\|b_i\|_2^2 \leq \phi \frac{K}{L}, i = 1, \ldots, L$ for some constant $\phi$.
\end{itemize}
The Riemannian steepest descent method that we use requires an initial iterate sufficiently close to the desired solution. Two neighborhoods of the desired solution in the quotient manifold $\mathbb{C}_*^{K} \times \mathbb{C}_*^{N} / \mathbb{C}_*$ are defined as
\begin{equation*}
\Omega_{\mu} = \{\pi(h, m) \mid \sqrt{L} \|B h\|_{\infty} \|m\|_2 \leq 4 d_* \mu\} \hbox{ and } \Pi_{\varepsilon} = \{\pi(h, m) \mid \|h m^* - h_{\sharp} m_{\sharp}^*\|_F \leq \varepsilon d_*\},
\end{equation*}
where $d_*$ and $\mu$ are defined in~\eqref{BD:e12} and $\varepsilon$ is assumed to be smaller than $1/15$ throughout this paper.
\subsection{Initialization} \label{BD:s21}

We consider the initialization by the spectral method in~\cite{LLSW2016}. The following algorithm and theorem have been given therein and we give them here in Algorithm~\ref{BD:a4} and Theorem~\ref{BD:th2} for completeness.

\begin{algorithm}
\caption{Initialization via spectral method and projection}
\label{BD:a4}
\begin{algorithmic}[1]
\ENSURE $(h_0, m_0)$ and $d$
\STATE Compute $\mathcal{A}^*(y)$;
\STATE Find the leading singular value $d$, left singular vector $\tilde{h}_0$ and right singular vector $\tilde{m}_0$ of $\mathcal{A}^*(y)$;
\STATE Find $h_0 = \arg\min_{z} \|z - \sqrt{d} \tilde{h}_0\|_2^2$ such that $\sqrt{L} \|Bz\|_\infty \leq 2 \sqrt{d} \mu$ and set $m_0 = \sqrt{d} \tilde{m}_0$;
\end{algorithmic}
\end{algorithm}
\begin{theo}[Initialization~\cite{LLSW2016}] \label{BD:th2}
The initialization obtained via Algorithm~\ref{BD:a4} satisfies\footnote{It follows from \cite[Theorem 3.1]{LLSW2016} that one representation in the orbit $[(h_0, m_0)]$ is in the neighborhood $\Omega_{\frac{1}{2} \mu} \cap \Pi_{\frac{2}{5} \varepsilon}$. Since $\Omega_{\frac{1}{2} \mu} \cap \Pi_{\frac{2}{5} \varepsilon}$ is independent of representation, \eqref{BD:e40} holds.}
\begin{equation} \label{BD:e40}
\pi(h_0, m_0) \in \Omega_{\frac{1}{2} \mu} \cap \Pi_{\frac{2}{5} \varepsilon}
\end{equation}
and
\begin{equation*}
\frac{9}{10} d_* \leq d \leq \frac{11}{10} d_*
\end{equation*}
holds with probability at least $1 - L^{-\gamma}$ if the number of measurements satisfies
\begin{equation*}
L \geq C_\gamma (\mu_h^2 + \sigma^2) \max(K, N) \log^2(L) / \varepsilon,
\end{equation*}
where $\varepsilon$ is any predetermined constant in $(0, \frac{1}{15}]$, and $C_\gamma$ is a constant only linearly depending on $\gamma$ with $\gamma \geq 1$.
\end{theo}

\subsection{Convergence Analysis} \label{BD:s5}

The proposed algorithm is stated in Algorithm~\ref{BD:a5}, which is an implementation of a Riemannian steepest descent method. Theorem~\ref{BD:th3} states that if the initial iterate is close enough to the desired solution and has sufficiently low coherence, then with high probability, the sequence $\{\pi(h_k, m_k)\}$ generated by Algorithm~\ref{BD:a5} converges to the desired solution linearly, up to noise. The proof of Theorem~\ref{BD:th3} is given in Appendices~\ref{BD:s8} and~\ref{BD:s7}.

\begin{algorithm}
\caption{An implementation of a Riemannian steepest descent method}
\label{BD:a5}
\begin{algorithmic}[1]
\STATE Given $h_0, m_0$, and set $k\leftarrow0$;
\FOR {$k = 0, 1, 2, \ldots$}
\STATE Set
\begin{equation} \label{BD:e38}
d_k = \|h_k\|_2 \|m_k\|_2, \;\; h_k \gets \sqrt{d_k} \frac{h_k}{\|h_k\|_2}; \;\; m_k \gets \sqrt{d_k} \frac{m_k}{\|m_2\|_2}
\end{equation}
and
\begin{equation} \label{BD:e37}
(h_{k+1}, m_{k+1}) = (h_k, m_k) - \frac{\alpha}{d_k} \left(\nabla_{h_k} \tilde{f}(h_k, m_k), \nabla_{m_k} \tilde{f}(h_k, m_k)\right)
\end{equation}
where $\nabla_h \tilde{f}(h, m) = J_1 m + \frac{L \rho}{4 d^2 \mu^2} \sum_{i = 1}^L G_0'\left(\frac{L|b_i^* h|^2 \|m\|_2^2}{8 d^2 \mu^2}\right) (b_i b_i^* h \|m\|_2^2)$ is the Euclidean gradient with respect to $h$, $\nabla_m \tilde{f}(h, m) = J_1^* h + \frac{L \rho}{4 d^2 \mu^2} \sum_{i = 1}^L G_0'\left(\frac{L|b_i^* h|^2 \|m\|_2^2}{8 d^2 \mu^2}\right) (m |b_i^* h|^2)$ is the Euclidean gradient with respect to $m$, and $J_1 = 2 \left(B^* \diag(\diag(B h m^* C^*) - y) {C}\right)$.
\ENDFOR
\end{algorithmic}
\end{algorithm}



\begin{theo} \label{BD:th3}
Suppose the initialization $\pi(h_0, m_0) \in \Omega_{\frac{1}{2} \mu} \cap \Pi_{\frac{2}{5} \varepsilon}$ and $L \geq C_\gamma (\mu^2 + \sigma^2) \max(K, N) \log^2(L) / \varepsilon^2$. Then the iterates generated by Algorithm~\ref{BD:a5} convergence linearly to $\pi(h_{\sharp}, m_{\sharp})$ in the sense that with probability at least $1 - 4 L^{-\gamma} - \frac{1}{\gamma} \exp(- (K + N))$, it holds that
\begin{equation*}
\|h_k m_k^* - h_{\sharp} m_{\sharp}^*\|_F \leq \frac{2}{3} \left(1 - \frac{\alpha a_0}{2}\right)^{k / 2} \varepsilon d_* + 50 \|\mathcal{A}^*(e)\|_2,
\end{equation*}
where $a_0 = 1 / 1500$, \whfirrev{}{$\alpha < 1 / (2 a_L)$ is a fixed step size, $a_L$ is the smoothness constant of the Riemannian objective as defined by Condition~\ref{BD:c4},}
$$
\|\mathcal{A}^*(e)\|_2 \leq \psi \sigma d_* \max\left(\sqrt{\frac{(\gamma+1) \max(K, N) \log(L)}{L}}, \frac{(\gamma+1)\sqrt{K N} \log^2(L)}{L}\right),
$$
and $\psi>0$ is a constant.
\end{theo}

\whfirrev{}{
The convergence rate of Algorithm~\ref{BD:a5} is determined by the step size $\alpha$. If $\mu^2 = O(\frac{L}{(K+N) \log^2 L})$ and $\rho \approx d^2 + 2.5 \|e\|_2^2$, then $a_L = O((1 + \sigma^2)(K + N) \log^2 L)$. See more details in Section~\ref{BD:s7} and~\cite[Page~13]{LLSW2016}.
}

We emphasize the differences between Algorithm~\ref{BD:a5} and the algorithm in~\cite{LLSW2016}. The steepest descent method in~\cite{LLSW2016} is stated in~Algorithm~\ref{BD:a2}, where $(\nabla_{h_k}^{w} \tilde{F}, \nabla_{m_k}^{w} \tilde{F})$ denotes the Wirtinger derivatives. 
It can be shown that the Wirtinger derivatives $\left(\nabla_{h_k}^{w} \tilde{F}(h_k, m_k), \nabla_{m_k}^{w} \tilde{F}(h_k, m_k)\right)$ and Euclidean gradient $\left(\nabla_{h_k} \tilde{F}(h_k, m_k), \nabla_{m_k} \tilde{F}(h_k, m_k)\right)$ satisfy
\begin{equation} \label{BD:e27}
\left(\nabla_{h_k}^{w} \tilde{F}(h_k, m_k), \nabla_{m_k}^{w} \tilde{F}(h_k, m_k)\right) = \frac{1}{2} \left(\nabla_{h_k} \tilde{F}(h_k, m_k), \nabla_{m_k} \tilde{F}(h_k, m_k)\right),
\end{equation}
which yields~\eqref{BD:e39}. The differences are given as follows:
\begin{enumerate}
\item The gradient of $\tilde{f}$ in Algorithm~\ref{BD:a5} is different from the gradient of $\tilde{F}$ in Algorithm~\ref{BD:a2} due to the differences in their penalty terms;
\item The coefficient of the Euclidean gradient in Algorithm~\ref{BD:a5} is related to the norms of $h_k$ and $m_k$ while the coefficient in Algorithm~\ref{BD:a2} is a constant;
\item The $h_k$ and $m_k$ are normalized in Algorithm~\ref{BD:a5} so that $h_k$ and $m_k$ have the same norm while the norms of $h_k$ and $m_k$ can be different \whfirrev{}{in Algorithm~\ref{BD:a2}}.
\end{enumerate}
\begin{algorithm}
\caption{The steepest descent method in~\cite[Algorithm~2]{LLSW2016}}
\label{BD:a2}
\begin{algorithmic}[1]
\STATE $k\leftarrow0$;
\FOR {$k = 0, 1, 2, \ldots$}
\STATE Set
\begin{align}
(h_{k+1}, m_{k+1}) =& (h_k, m_k) - \alpha (\nabla_{h_k}^w \tilde{F}(h_k, m_k), \nabla_{m_k}^w \tilde{F}(h_k, m_k)) \nonumber \\
=& (h_k, m_k) - \frac{\alpha}{2} (\nabla_{h_k} \tilde{F}(h_k, m_k), \nabla_{m_k} \tilde{F}(h_k, m_k)). \label{BD:e39}
\end{align}
\ENDFOR
\end{algorithmic}
\end{algorithm}

\whfirrev{}{In order to understand the behavior of a generic Riemannian optimization algorithm, such as Newton method, quasi-Newton method, CG method etc, in a small neighborhood of the solution, we give the following remark, which will be formally stated and proven in Theorem~\ref{BD:th4} after we introduce appropriate Riemannian tools.}
\begin{rema} \label{BD:th5}
With high probability provided $L$ large enough, i) a Riemannian version of the Hessian of $f$ in~\eqref{BD:e14} at $\pi(h_{\sharp}, m_{\sharp})$ is positive definite, ii) there exists a lower bound, not dependent on $K$ and $N$, of the smallest eigenvalue of the Riemannian Hessian, 
and iii) the condition number of the Riemannian Hessian at $\pi(h_{\sharp}, m_{\sharp})$ is small.
\end{rema}
The results in Remark~\ref{BD:th5} have important implications.
Positive definiteness of the Hessian at $\pi(h_{\sharp}, m_{\sharp})$ implies that there exists a neighborhood of $\pi(h_{\sharp}, m_{\sharp})$ in which $f$ has a unique minimizer. Therefore, if an optimization algorithm generates iterates that converge and stay in the neighborhood, then the limit must be $\pi(h_{\sharp}, m_{\sharp})$. The independence of the lower bound of the smallest eigenvalue from $K$ and $N$ indicates that it is worthwhile investigating whether the size of the neighborhood depends on $K$ and $N$.
\whfirrev{}{Note that it is unknown if the neighborhood is in the basin of attraction $\Omega_{\frac{1}{2} \mu} \cap \Pi_{\frac{2}{5} \varepsilon}$ in Theorem~\ref{BD:th3}.} We leave this work for future research. The good condition number of the Hessian at the minimizer implies that a Riemannian steepest descent method can converge quickly.

\section{Riemannian Optimization} \label{BD:s12}

Recently many Riemannian optimization methods have been systemically analyzed and efficient libraries have been designed, e.g., Riemannian trust-region Newton method (RTR-Newton) \cite{BAKER08}, Riemannian Broyden family method including BFGS method and its limited-memory version (RBroyden family, RBFGS, LRBFGS) \cite{RinWir2012,HUANG2013,HGA2014,HuaAbsGal2018}, Riemannian trust-region symmetric rank-one update method and its limited-memory version (RTR-SR1, LRTR-SR1) \cite{HUANG2013,HAG13}, Riemannian Newton method (RNewton) and Riemannian non-linear conjugate gradient method (RCG) \cite{AMS2008,SI2015,Sato2015,Zhu2016}.

\whfirrev{}{This section follows the optimization framework on quotient manifolds in~\cite{AMS2008} and uses it to derive the tools for optimization problems on $\mathbb{C}_*^K \times \mathbb{C}_*^N / \mathbb{C}_*$.} Specifically, we describe the representation of the manifold $\mathbb{C}_r^{K \times N}$ for $r \geq 1$, and the ingredients which are required in Riemannian optimization, such as tangent space, retraction, vector transport and gradient.
We use only $r = 1$ in this paper for the blind deconvolution problem. The results for $r > 1$ are also provided since deriving them does not require much more work and is interesting for future research.\footnote{The following provides an example of a problem where taking $r > 1$ is useful. As shown in~\cite{HGZ2017}, it is known that the desired solution for the phase retrieval problem by a lifting approach is a rank-1 matrix. If the number of measurements is small, then optimizing over a fixed-rank manifold with rank greater than~1 is empirically shown to be more efficient and more effective than $r = 1$.} 

\subsection{Fixed-rank Manifold and Two-Factor Representation} \label{BD:s4}

A point on a quotient manifold is an equivalence class, which is often cumbersome computationally. In practice, choosing a representative for an equivalence class and definitions of related mathematical objects have been developed in many papers in the literature of computation on manifolds, e.g.,~\cite{AMS2008}.

Let $\mathcal{N}_r$ denote the total space $\mathbb{C}_*^{K \times r} \times \mathbb{C}_*^{N \times r}$ and $\mathcal{Q}_r$ denote the quotient space $\mathcal{N}_r / \GL(r)$. 
\whfirrev{}{Define a function $\pi: \mathcal{N}_r \rightarrow \mathcal{Q}_r: (H, M) \mapsto \pi(H, M)$.} Let $\pi(H, M)$ denote $[(H, M)] = \{(H P^{-1}, M P^*) \mid P \in \GL(r)\}$ viewed as an element of $\mathcal{Q}_r$ and $\pi^{-1}(\pi(H, M))$ is used to denote $[(H, M)]$ viewed as a subset of $\mathcal{N}_r$.
Let $\T_{(H, M)} \mathcal{N}_r$ denote the tangent space of $\mathcal{N}_r$ at $(H, M)$.
Given $\zeta_{(H, M)} \in \T_{(H, M)} \mathcal{N}_r$, denote the first and second components of $\zeta_{(H, M)}$ by $\zeta_H$ and $\zeta_M$, i.e., $\zeta_{(H, M)} = (\zeta_H, \zeta_M)$, where $\zeta_H \in \mathbb{C}^{K \times r}$ and $\zeta_M \in \mathbb{C}^{N \times r}$.
The vertical space at a point $(H, M) \in \pi^{-1}(\pi(H, M))$ is defined to be the tangent space of $\pi^{-1}(\pi(H, M))$ at $(H, M)$. That is  
\begin{equation*}
\mathcal{V}_{(H, M)} = \{(-H \Lambda, M \Lambda^*) \mid \Lambda \in \mathbb{C}^{r \times r}\}.
\end{equation*}
The Riemannian metric of $\mathcal{N}_r$ that we use is
\begin{equation} \label{BD:Metric}
g_{(H, M)}\Big(\eta_{(H, M)}, \xi_{(H, M)}\Big) = \mathrm{Re}\left(\trace\left(\eta_H^* \xi_H (M^* M)+\eta_M^* \xi_M (H^* H)\right)\right),
\end{equation}
where $\eta_{(H, M)}, \xi_{(H, M)} \in \T_{(H, M)} \mathcal{N}_r$.
It follows that the induced norm is
$$
\|\eta_{(H, M)}\|_{g_{(H, M)}} = \sqrt{g_{(H, M)}\Big(\eta_{(H, M)}, \eta_{(H, M)}\Big)}.
$$
We remove the subscript $(H, M)$ of $g$ in $\|\eta_{(H, M)}\|_{g_{(H, M)}}$, i.e., use $\|\eta_{(H, M)}\|_{g}$, when there is no confusion.

\whfirrev{}{
The idea of using the Riemannian metric~\eqref{BD:Metric} is not new. A similar metric on the real fixed-rank manifold has been used in~\cite{Mishra2014}. The Euclidean metric $\inner[2]{\eta_{(H, M)}}{\xi_{(H, M)}} = \mathrm{Re}\left(\trace\left(\eta_H^* \xi_H + \eta_M^* \xi_M\right)\right)$ on $\mathcal{N}_r$ is not used here since it is not invariant to the representative in $\pi^{-1}(\pi(H, M))$. This can be verified by Lemma~\ref{BD:le1}.
The Riemannian metric~\eqref{BD:Metric} is not equivalent to the Euclidean metric in $\mathbb{C}^{K \times N}$, i.e.,
\begin{equation*}
g_{(H, M)}\Big(\eta_{(H, M)}, \xi_{(H, M)}\Big) \neq \mathrm{Re}\left(\trace\left((\eta_H M^* + H \eta_M^*)^* (\xi_H M^* + H \xi_M^*) \right) \right)
\end{equation*}
in general. \whfirrev{}{Note that the tangent vector $(\eta_H, \eta_M)$ of $\mathcal{N}_r$ at $(H, M)$ corresponds to the tangent vector $\eta_H M^* + H \eta_M^*$ of $\mathbb{C}^{K \times N}$ at $H M^*$.}
}

The horizontal space at $(H, M)$, denoted by $\mathcal{H}_{(H, M)}$, is defined to be the subspace of $\T_{(H, M)} \mathcal{N}_r$ that is orthogonal to $\mathcal{V}_{(H, M)}$ with respect to the metric~\eqref{BD:Metric}. That is
\begin{equation*}
\mathcal{H}_{(H, M)} = \left\{
\left(
\begin{bmatrix}
H & H_\perp
\end{bmatrix}
\begin{bmatrix}
K \\
T
\end{bmatrix}
,
\begin{bmatrix}
M & M_\perp
\end{bmatrix}
\begin{bmatrix}
K^* \\
Q
\end{bmatrix}
\right)
\mid K \in \mathbb{C}^{r \times r}, T \in \mathbb{C}^{(K - r) \times r}, Q \in \mathbb{C}^{(N - r) \times r}
\right\},
\end{equation*}
where $H_\perp$ denotes a $K \times (K-r)$ orthonormal matrix such that $H^* H_\perp = 0$ and likewise for $M_\perp$.
The horizontal space $\mathcal{H}_{(H, M)}$ is a representation of the tangent space $\T_{\pi(H, M)} \mathcal{Q}_r$. It is known that for any $\eta_{\pi(H, M)} \in \T_{\pi(H, M)} \mathcal{Q}_r$, there exists a unique vector in $\mathcal{H}_{(H, M)}$, called the horizontal lift of $\eta_{\pi(H, M)}$ and denoted by $\eta_{\uparrow_{(H, M)}}$, satisfying $\D \pi(H, M) [\eta_{\uparrow_{(H, M)}}] = \eta_{\pi(H, M)}$, see e.g.,~\cite{AMS2008}.
The relationship among horizontal lifts of a tangent vector $\eta_{\pi(H, M)}$ is given in Lemma~\ref{BD:le1}. The result follows from~\cite[Theorem 9.3.1]{HUANG2013}.
\begin{lemm} \label{BD:le1}
A vector field $(\hat{\theta}, \hat{\vartheta})$ on $\mathcal{N}_r$ is the horizontal lift of a vector field on $\mathcal{Q}_r$ if and only if, for each $(H, M) \in \mathcal{N}_r$, we have
\begin{equation}\label{BD:e29}
\left(\hat{\theta}_{H P^{-1}}, \hat{\vartheta}_{M P^*}\right) = \left(\hat{\theta}_H P^{-1}, \hat{\vartheta}_M P^*\right),
\end{equation}
for all $P \in \GL(r)$.
\end{lemm}
The Riemannian metric~\eqref{BD:Metric} on $\mathcal{N}_r$ defines a Riemannian metric on~$\mathcal{Q}_r$.
\begin{lemm}
The following defines a Riemannian metric on $\mathcal{Q}_r$:
\begin{equation} \label{BD:QMetric}
g_{\pi(H, M)}\Big(\eta_{\pi(H, M)}, \xi_{\pi(H, M)}\Big) = \mathrm{Re}\left(\trace\left(\eta_{\uparrow_H}^* \xi_{\uparrow_H} (M^* {M}) + \eta_{\uparrow_M}^* \xi_{\uparrow_M} (H^* {H})\right)\right).
\end{equation}
\end{lemm}
\begin{proof}
It can be easily verified that the right hand side of~\eqref{BD:QMetric} is invariant to representions in $\pi^{-1}(\pi(H, M))$.
\end{proof}


The orthogonal projections onto the horizontal space or the vertical space are similar to those in~\cite[Table 4.3]{Mishra2014}. We give them below for completeness.
\begin{lemm} \label{BD:le8}
The orthogonal projection to the vertical space $\mathcal{V}_{(H, M)}$ is $P_{(H, M)}^v\left(\eta_{(H, M)}\right) = (- H \Lambda, M \Lambda^*)$, where $\Lambda = 0.5 (\eta_M^* {M} (M^* M)^{-1} - (H^* H)^{-1} H^* \eta_H)$. The orthogonal projection to the horizontal space $\mathcal{H}_{(H, M)}$ is $P_{(H, M)}^h\left(\eta_{(H, M)}\right) = \eta_{(H, M)} - P_{(H, M)}^v\left(\eta_{(H, M)}\right)$.
\end{lemm}
\begin{proof}
It is easy to verify that $g\left(P_{(H, M)}^v\left(\eta_{(H, M)}\right), P_{(H, M)}^h\left(\eta_{(H, M)}\right)\right) = 0$ and $P_{(H, M)}^v\left(\eta_{(H, M)}\right) + P_{(H, M)}^h\left(\eta_{(H, M)}\right) = \eta_{(H, M)}$.
\end{proof}

\subsection{Retraction and Vector transport}

Retraction is used to update iterates in a Riemannian algorithm and vector transport is used to compare tangent vectors in different tangent spaces. Specifically, a retraction $R$ is a smooth mapping from the tangent bundle $\T \mathcal{M}$, which is the set of all tangent vectors to $\mathcal{M}$, onto a manifold $\mathcal{M}$ such that (i) $R(0_x)=x$ for all $x\in\mathcal{M}$ (where $0_x$ denotes the origin of $\T_x \mathcal{M}$) and (ii) $\frac{d}{dt}R(t\xi_x)|_{t=0} = \xi_x$ for all $\xi_x\in\T_x\mathcal{M}$. The restriction of $R$ to $\T_x\mathcal{M}$ is denoted by $R_x$. A vector transport $\mathcal{T}: \T\mathcal{M}\oplus\T\mathcal{M} \to \T \mathcal{M}, (\eta_x,\xi_x)\mapsto \mathcal{T}_{\eta_x} \xi_x$ with associated retraction $R$ is a smooth mapping such that, for all $(x, \eta_x)$ in the domain of $R$ and all $\xi_x \in \T_x\mathcal{M}$, it holds that (i) $\mathcal{T}_{{\eta_x}}\xi_x \in \T_{R(\eta_x)}\mathcal{M}$, (ii) $\mathcal{T}_{{0_x}}\xi_x = \xi_x$, (iii) $\mathcal{T}_{{\eta_x}}$ is a linear map.
A retraction in the total space $\mathcal{N}_r$ is
\begin{equation}\label{BD:e1}
R_{(H, M)}(\eta_{(H, M)}) = (H + \eta_{H}, M + \eta_{M}).
\end{equation}
It follows from Lemma~\ref{BD:le1} that $\pi(R_{(H, M)}(\eta_{(H, M)})) = \pi(R_{(HP^{-1}, M P^*)}(\eta_{(HP^{-1}, MP^*)}))$. Therefore, \eqref{BD:e1} defines a retraction in the quotient manifold $\mathcal{Q}_r$, i.e.,
\begin{equation} \label{BD:e34}
\tilde{R}_{\pi(H, M)} (\eta_{\pi(H, M)}) := \pi(R_{(H, M)}(\eta_{\uparrow_{(H, M)}})).
\end{equation}
The vector transport used is the vector transport by parallelization \cite{HAG2016VT}:
\begin{equation} \label{BD:e19}
\mathcal{T}_{\eta_x} = \mathcal{B}_y \mathcal{B}_x^\dag,
\end{equation}
where $y = R_x(\eta_x)$, $\mathcal{B}$ is a smooth tangent basis field defined on an open set of a manifold $\mathcal{M}$ and $\mathcal{B}_x^\dag$ denotes the pseudo-inverse of $\mathcal{B}_x$.
A smooth orthonormal basis of $\mathcal{H}_{(H, M)}$ is 
\begin{gather*}
\left\{\left( \frac{1}{\sqrt{2}} H L_H^{-*} e_i e_j^T L_M^{-1}, \frac{1}{\sqrt{2}} M L_M^{-*} e_j e_i^T L_H^{-1} \right), i = 1, \ldots, r, j = 1, \ldots, r \right\} \\
\bigcup \left\{\left( \frac{1}{\sqrt{2}} H L_H^{-*} e_i e_j^T L_M^{-1} \sqrt{-1}, -\frac{1}{\sqrt{2}} M L_M^{-*} e_j e_i^T L_H^{-1} \sqrt{-1} \right), i = 1, \ldots, r, j = 1, \ldots, r \right\} \\
\bigcup \left\{ \left(H_\perp \tilde{e}_i e_j^T L_M^{-1}, 0_{N \times r} \right), i = 1, \ldots K - r, j = 1, \ldots r \right\} \\
\bigcup \left\{ \left(H_\perp \tilde{e}_i e_j^T L_M^{-1} \sqrt{-1}, 0_{N \times r} \right), i = 1, \ldots K - r, j = 1, \ldots r \right\} \\
\bigcup \left\{ \left(0_{K \times r}, M_\perp \hat{e}_i e_j^T L_H^{-1} \right), i = 1, \ldots N - r, j = 1, \ldots r \right\} \\
\bigcup \left\{ \left(0_{K \times r}, M_\perp \hat{e}_i e_j^T L_H^{-1} \sqrt{-1} \right), i = 1, \ldots N - r, j = 1, \ldots r \right\},
\end{gather*}
where $(e_1, \ldots, e_r)$ is the canonical basis of $\mathbb{R}^r$, $(\tilde{e}_1, \ldots, \tilde{e}_{(K - r)})$ is the canonical basis of $\mathbb{R}^{K - r}$, $(\hat{e}_1, \ldots, \hat{e}_{(N - r)})$ is the canonical basis of $\mathbb{R}^{N - r}$, and $H^* H = L_H L_H^*$ and $M^* M = L_M L_M^*$ are Cholesky decompositions. Therefore, the vector transport~\eqref{BD:e19} with the above orthonormal basis yields locally smooth, linear mapping from $\mathcal{H}_{(H, M)}$ to $\mathcal{H}_{R_{(H, M)}(\eta_{(H, M)})}$. In the case of $r = 1$, one can choose $h_\perp$ to be the last $K-1$ columns of the Householder matrix $I - 2 v v^* / \|v\|_2^2$, where $v = h - \|h\| e_1$, and similarly for $m_\perp$. It follows that any corresponding pair of tangent vectors $\eta_{(h, m)}$ and $(\eta_{hp^{-1}}, \eta_{mp^*})$ in bases $B_{(h, m)}$ and $B_{(h p^{-1}, m p^*)}$ satisfies~\eqref{BD:e29}, i.e., $(\eta_{hp^{-1}}, \eta_{mp^*}) = (\eta_h p^{-1}, \eta_m p^*)$. Therefore, the mapping
\begin{equation} \label{BD:e30}
\left(\mathcal{T}_{\eta_{\pi(h, m)}} \xi_{\pi(h, m)} \right)_{\uparrow_{(h, m)}} := B_{(\tilde{h}, \tilde{m})} B_{(h, m)}^\dagger
\end{equation}
defines a vector transport on $\mathcal{Q}_1$, where $(\tilde{h}, \tilde{m}) = R_{(h, m)}\left(\eta_{\uparrow_{(h, m)}}\right)$.

\subsection{Riemannian Gradient}

The Riemannian gradient of a function $\mathbf{f}: \mathcal{M} \to \mathbb{R}: x \mapsto \mathbf{f}(x)$ is defined to be the unique tangent vector $\grad \mathbf{f} \in \T_x \mathcal{M}$ satisfying
\begin{equation*}
\mathrm{D} \mathbf{f}(x) [\eta_x] = g(\grad \mathbf{f}(x), \eta_x), \;\; \forall \eta_x \in \T_x \mathcal{M},
\end{equation*}
where $g$ is the Riemannian metric of $\mathcal{M}$. The Riemannian gradient of~\eqref{BD:e13} is derived in Lemma~\ref{BD:le14}.
\begin{lemm} \label{BD:le14}
Given any $\pi(h, m) \in \mathcal{Q}_1$, the gradient of $\tilde{f}_{\pi(h, m)}$ is
\begin{gather} \label{BD:e18}
\left(\grad \tilde{f}(\pi(h, m))\right)_{\uparrow_{(h, m)}} = P_{(h, m)}^h \Big(\nabla_h \tilde{f}(h, m) (m^* m)^{-1}, \nabla_m \tilde{f}(h, m) (h^* h)^{-1} \Big),
\end{gather}
where
$
\nabla_h \tilde{f}(h, m) = J_1 m + \frac{L \rho}{4 d^2 \mu^2} \sum_{i = 1}^L G_0'\left(\frac{L|b_i^* h|^2 \|m\|_2^2}{8 d^2 \mu^2}\right) (b_i b_i^* h \|m\|_2^2)$ is the Euclidean gradient with respect to $h$,
$\nabla_m \tilde{f}(h, m) = J_1^* h + \frac{L \rho}{4 d^2 \mu^2} \sum_{i = 1}^L G_0'\left(\frac{L|b_i^* h|^2 \|m\|_2^2}{8 d^2 \mu^2}\right) (m |b_i^* h|^2)$ is the Euclidean gradient with respect to $m$, and $J_1 = 2 \left(B^* \diag(\diag(B h m^* C^*) - y) {C}\right)$. Moreover, if the representation $(h, m)$ satisfies $\|h\|_2 = \|m\|_2$, then
\begin{gather} \label{BD:e28}
\left(\grad \tilde{f}(\pi(h, m))\right)_{\uparrow_{(h, m)}} = \Big(\nabla_h \tilde{f}(h, m) (m^* m)^{-1}, \nabla_m \tilde{f}(h, m) (h^* h)^{-1} \Big).
\end{gather}
\end{lemm}
\begin{proof}
Taking directional derivative along any direction $\eta_{\pi(h, m)} \in \T_{\pi(h, m)} \mathcal{Q}_r$ yields
\begin{align*}
\D \tilde{f}(\pi(h, m)) [\eta_{\pi(h, m)}] =& \D \left( \|y - \diag(B h m^* C^*)\|_2^2 + \rho \sum_{i = 1}^L G_0\left(\frac{L |b_i^* h|^2 \|m\|_2^2}{8 d^2 \mu^2}\right) \right) \left[ \left( \eta_{\uparrow_h}, \eta_{\uparrow_m}\right) \right] \\
=& 2 \mathrm{Re} \trace\left((\diag(B h m^* C^*) - y)^T \diag(B \eta_{\uparrow_h} m^* C^*)\right) \\
&+ 2 \mathrm{Re} \trace\left((\diag(B h m^* C^*) - y)^T \diag(B h \eta_{\uparrow_m}^* C^*)\right) \\
&+ \mathrm{Re} \left( \frac{L \rho}{4 d^2 \mu^2} \sum_{i = 1}^L G_0'\left(\frac{L|b_i^* h|^2 \|m\|_2^2}{8 d^2 \mu^2}\right) (\eta_{\uparrow_h}^* b_i b_i^* h \|m\|_2^2) \right) \\
&+ \mathrm{Re} \left( \frac{L \rho}{4 d^2 \mu^2} \sum_{i = 1}^L G_0'\left(\frac{L|b_i^* h|^2 \|m\|_2^2}{8 d^2 \mu^2}\right) (\eta_{\uparrow_m}^* m |b_i^* h|^2) \right) \\
=& \mathrm{Re} \left(\eta_{\uparrow_h}^* \nabla_h \tilde{f}(h, m) + \eta_{\uparrow_m}^* \nabla_m \tilde{f}(h, m)\right).
\end{align*}
Combining the above equation with the definition~\cite[(3.31)]{AMS2008} and~\eqref{BD:QMetric} yields~\eqref{BD:e18}. If $\|h\|_2 = \|m\|_2$, then the projection is not necessary since $\Big(\nabla_h \tilde{f}(h, m) (m^* m)^{-1}, \nabla_m \tilde{f}(h, m) (h^* h)^{-1} \Big)$ is in $\mathcal{H}_{(h, m)}$ already.
\end{proof}

%

\subsection{A Riemannian Steepest Descent Algorithm} \label{BD:s13}

Algorithm~\ref{BD:a5} is an instance of a Riemannian steepest descent method.
A Riemannian steepest descent iteration for optimizing a real-valued function $\mathbf{f}$ on $\mathcal{M}$ is $x_{k+1} = R_{x_k}(- \alpha \grad \mathbf{f}(x_k))$, where $\alpha>0$ is a step size and $R$ is a retraction. In the case of optimizing $\tilde{f}$ in~\eqref{BD:e13}, the iteration is $\pi(h_{(k + 1)}, m_{(k + 1)}) = R_{\pi(h_k, m_k)} (\grad \tilde{f}(\pi(h_k, m_k)))$, which yields implementations stated in Algorithm~\ref{BD:a1}. The update \eqref{BD:e31} is the generic Riemannian steepest descent update formula and~\eqref{BD:e32} is for this specifical problem. Since the iteration $\pi(h_{(k + 1)}, m_{(k + 1)}) = R_{\pi(h_k, m_k)} (\grad \tilde{f}(\pi(h_k, m_k)))$ is independent of representations chosen in equivalence class, we can choose $h_k$ and $m_k$ such that they always have the same norm.
In addition, the projection onto the horizontal space is not necessary due to~\eqref{BD:e28}. It follow that we have the update formula~\eqref{BD:e33}, which is used in Algorithm~\ref{BD:a5}. 


\begin{algorithm}
\caption{A Riemannian steepest descent method}
\label{BD:a1}
\begin{algorithmic}[1]
\STATE $k\leftarrow0$;
\FOR {$k = 0, 1, 2, \ldots$}
\STATE Set
\begin{equation}\label{BD:e31}
(h_{k+1}, m_{k+1}) = \tilde{R}_{(h_k, m_k)}\left(- \alpha \left(\grad \tilde{f}(\pi(h, m))\right)_{\uparrow_{(h, m)}} \right);
\end{equation}
or equivalently
\begin{equation}\label{BD:e32}
(h_{k+1}, m_{k+1}) = (h_k, m_k) - \alpha P_{(h, m)}^h\left(\nabla_{h_k} \tilde{f}(h_k, m_k) (m_k^* m_k)^{-1}, \nabla_{m_k} \tilde{f}(h_k, m_k) (h_k^* h_k)^{-1}\right);
\end{equation}
or equivalently $d_k = \|h_k\|_2 \|m_k\|_2$, $h_k \gets \sqrt{d_k} \frac{h_k}{\|h_k\|_2}$; $m_k \gets \sqrt{d_k} \frac{m_k}{\|m_2\|_2}$ and
\begin{equation} \label{BD:e33}
(h_{k+1}, m_{k+1}) = (h_k, m_k) - \frac{\alpha}{d_k} \left(\nabla_{h_k} \tilde{f}(h_k, m_k), \nabla_{m_k} \tilde{f}(h_k, m_k)\right);
\end{equation}
\ENDFOR
\end{algorithmic}
\end{algorithm}

\subsection{Local Convexity}  \label{BD:s6}

In this section, we give the complete version of Remark~\ref{BD:th5} in Theorem~\ref{BD:th4}. We consider the Riemannian Hessian of the function $f \circ \tilde{R}_{\pi(h_{\sharp}, m_{\sharp})}$, which is defined on the tangent space $\T_{\pi(h_{\sharp}, m_{\sharp})} \mathcal{Q}_1$. Since a tangent space is a vector space, the Riemannian Hessian with respect to the metric~\eqref{BD:QMetric} is
\begin{align*}
&\left(\Hess f \circ \tilde{R}_{\pi(h, m)}(\xi_{\pi(h, m)})  [\eta_{\pi(h, m)}]\right)_{\uparrow_{(h, m)}} \\
= &P_{(h, m)}^h \left((J_3 (m + \xi_{\uparrow_m}) + J_1 \eta_{\uparrow_{m}}) (m^* m)^{-1}, (J_3^* (h + \xi_{\uparrow_h}) + J_1^* \eta_{\uparrow_{h}}) (h^* h)^{-1} \right),
\end{align*}
where $J_1$ is defined in~\eqref{BD:e18} and $J_3 = 2 \mathrm{Re}\left(B^* \diag(\diag(B\left(\eta_{\uparrow_{h}} (m + \eta_{\uparrow_{m}})^* + (h + + \xi_{\uparrow_h}) \eta_{\uparrow_{m}}^* \right)C^*)) C\right)$. 
The complete version of Remark~\ref{BD:th5} is given below.
\begin{theo} \label{BD:th4}
With probability at least $1 - L^{-\gamma}$, the eigenvalues of the Hessian of $f \circ \tilde{R}_{\pi(h_{\sharp}, m_{\sharp})}$ at $0_{\pi(h_{\sharp}, m_{\sharp})}$ are between $9 d_*^2 / 5$ and $22 d_*^2 / 5$ up to noise, i.e.,
\begin{equation*}
\frac{22 d_*^2}{5} + 4 \|\mathcal{A}^*(e)\|_2 d_* \geq \frac{g\left(\eta_{\pi(h_{\sharp}, m_{\sharp})}, \Hess f \circ \tilde{R}_{\pi(h_{\sharp}, m_{\sharp})}(0_{\pi(h_{\sharp}, m_{\sharp})})  [\eta_{\pi(h_{\sharp}, m_{\sharp})}]\right)}{g(\eta_{\pi(h_{\sharp}, m_{\sharp})}, \eta_{\pi(h_{\sharp}, m_{\sharp})})} \geq \frac{9 d_*^2}{5} - 4 \|\mathcal{A}^*(e)\|_2 d_*,
\end{equation*}
for all $\eta_{\pi(h_{\sharp}, m_{\sharp})} \in \T_{\pi(h_{\sharp}, m_{\sharp})} \mathcal{Q}_r$, where $\tilde{R}$ is defined in~\eqref{BD:e34}.
\end{theo}
\begin{proof}
Without loss of generality, we assume $\|h_{\sharp}\|_2 = \|m_{\sharp}\|_2 = \sqrt{d_*}$.
We have
\begin{align*}
&g\left(\eta_{\pi(h_{\sharp}, m_{\sharp})}, \Hess f \circ \tilde{R}_{\pi(h_{\sharp}, m_{\sharp})}(0_{\pi(h_{\sharp}, m_{\sharp})})  [\eta_{\pi(h_{\sharp}, m_{\sharp})}] \right) = \frac{d^2}{d t^2} f \circ \tilde{R}_{\pi(h_{\sharp}, m_{\sharp})} (t \eta_{\pi(h_{\sharp}, m_{\sharp})}) \vert_{t = 0} \\
=& \frac{d^2}{d t^2} \|y - \diag(B (h_{\sharp} + t \eta_{\uparrow_{h_{\sharp}}}) (m_{\sharp} + t \eta_{\uparrow_{m_{\sharp}}})^* C^*)\|_2^2 \vert_{t = 0} \\
=& \mathrm{Re}\left(\trace\left(\eta_{\uparrow_{h_{\sharp}}}^* (J_2 m_{\sharp} + J_1 \eta_{\uparrow_{m_{\sharp}}}) + \eta_{\uparrow_{m_{\sharp}}^*} (J_2^* h_{\sharp} + J_1^* \eta_{\uparrow_{h_{\sharp}}})\right)\right),
\end{align*}
where $J_1$ is defined in~\eqref{BD:e18} and $J_2 = 2 \mathrm{Re}\left(B^* \diag(\diag(B\left(\eta_{\uparrow_{h_{\sharp}}} m^* + h \eta_{\uparrow_{m_{\sharp}}}^* \right)C^*)) C\right)$. It follows from the definition of $\pi(h_{\sharp}, m_{\sharp})$ that $J_1 = 2 \mathcal{A}^*(e)$. Thus,
\begin{align}
&\frac{d^2}{d t^2} f \circ \tilde{R}_{\pi(h_{\sharp}, m_{\sharp})} (t \eta_{\pi(h_{\sharp}, m_{\sharp})}) \vert_{t = 0} = \mathrm{Re}(\trace(2 \eta_{\uparrow_{h_{\sharp}}} J_1 \eta_{\uparrow_{m_{\sharp}}}+ J_2^*(\eta_{\uparrow_{h_{\sharp}}} m^* + h_{\sharp} \eta_{\uparrow_{m_{\sharp}}}^*))) \nonumber \\
=& 4 \mathrm{Re} \trace\left( \eta_{\uparrow_{h_{\sharp}}}^* \mathcal{A}^*(e)\eta_{\uparrow_{m_{\sharp}}}\right) + 2 \|\mathcal{A}(\eta_{\uparrow_{h_{\sharp}}} m^* + h_{\sharp} \eta_{\uparrow_{m_{\sharp}}}^*)\|_2^2. \label{BD:e22}
\end{align}
It holds that
\begin{align}
|4 \mathrm{Re} \trace\left(\eta_{\uparrow_{h_{\sharp}}}^* \mathcal{A}^*(e)\eta_{\uparrow_{m_{\sharp}}} \right)| \leq& 4 \|\mathcal{A}^*(e)\|_2 \|\eta_{\uparrow_{m_{\sharp}}}\|_2 \|\eta_{\uparrow_{h_{\sharp}}}\|_2 \leq 4 \|\mathcal{A}^*(e)\|_2 \left(\|\eta_{\uparrow_{m_{\sharp}}}\|_2^2 + \|\eta_{\uparrow_{h_{\sharp}}}^2\|_2 \right) \nonumber \\
\leq& 4 \|\mathcal{A}^*(e)\|_2 d_* g(\eta_{\pi(h_{\sharp}, m_{\sharp})}, \eta_{\pi(h_{\sharp}, m_{\sharp})}). \label{BD:e20}
\end{align}
Using~\cite[Lemma~5.12]{LLSW2016} yields
\begin{equation}
\frac{11}{5} \|\eta_{\uparrow_{h_{\sharp}}} m_{\sharp}^* + h_{\sharp} \eta_{\uparrow_{m_{\sharp}}}^*\|_F^2 \geq 2 \|\mathcal{A}(\eta_{\uparrow_{h_{\sharp}}} m^* + h_{\sharp} \eta_{\uparrow_{m_{\sharp}}}^*)\|_2^2 \geq \frac{9}{5} \|\eta_{\uparrow_{h_{\sharp}}} m_{\sharp}^* + h_{\sharp} \eta_{\uparrow_{m_{\sharp}}}^*\|_F^2, \label{BD:e23}
\end{equation}
with probability at least $1 - L^{-\gamma}$ provided $L \geq C_\gamma \max(K, \mu_h^2 N) \log^2(L)$. Note that this probability is independent of $\eta_{\pi(h_{\sharp}, m_{\sharp})}$.

Decompose $\eta_{\uparrow_{h_{\sharp}}} = a h_{\sharp} + b \eta_{\uparrow_{h_\perp}}$ and $\eta_{\uparrow_{m_{\sharp}}} = a^* m_{\sharp} + \beta \eta_{\uparrow_{m_\perp}}$, where $h_{\sharp}^* \eta_{\uparrow_{h_\perp}} = 0$, $\|\eta_{\uparrow_{h_\perp}}\|_2=1$, $m_{\sharp}^* \eta_{\uparrow_{m_\perp}} = 0$, and $\|\eta_{\uparrow_{m_\perp}}\|_2=1$. Therefore, we have
\begin{align}
\|\eta_{\uparrow_{h_{\sharp}}} m^* + h_{\sharp} \eta_{\uparrow_{m_{\sharp}}}^*\|_F^2 =& \|2 a h_{\sharp} m_{\sharp}^*\|_F^2 + \|b \eta_{\uparrow_{h_\perp}} m_{\sharp}^*\|_F^2 + \|\beta h_{\sharp} \eta_{\uparrow_{m_\perp}}^*\|_F^2 = |2a|^2 d_*^2 + b^2 d_* + \beta^2 d_* \nonumber \\
\geq& 2 |a|^2 d_*^2 + b^2 d_* + \beta^2 d_* = (|a|^2 d_* + b^2 + |a|^2 d_* + \beta^2) d_* \nonumber \\
=& d_* (\|\eta_{\uparrow_{h_{\sharp}}}\|_2^2 + \|\eta_{\uparrow_{m_{\sharp}}}\|_2^2) = d_*^2 g(\eta_{\pi(h_{\sharp}, m_{\sharp})}, \eta_{\pi(h_{\sharp}, m_{\sharp})}). \label{BD:e24}
\end{align}
Similarly, we have
\begin{equation} \label{BD:e25}
\|\eta_{\uparrow_{h_{\sharp}}} m^* + h_{\sharp} \eta_{\uparrow_{m_{\sharp}}}^*\|_F^2 = |2a|^2 d_*^2 + b^2 d_* + \beta^2 d_* \leq 2 (|a|^2 d_* + b^2 + |a|^2 d_* + \beta^2) d_* \leq 2 d_*^2 g(\eta_{\pi(h_{\sharp}, m_{\sharp})}, \eta_{\pi(h_{\sharp}, m_{\sharp})}).
\end{equation}
Combining~\eqref{BD:e22}, \eqref{BD:e20}, \eqref{BD:e23}, \eqref{BD:e24} and~\eqref{BD:e25} yields
\begin{equation*}
\frac{22 d_*^2}{5} + 4 \|\mathcal{A}^*(e)\|_2 d_* \geq \frac{g\left(\eta_{\pi(h_{\sharp}, m_{\sharp})}, \Hess f \circ \tilde{R}_{\pi(h_{\sharp}, m_{\sharp})}(0_{\pi(h_{\sharp}, m_{\sharp})})  [\eta_{\pi(h_{\sharp}, m_{\sharp})}] \right)}{g(\eta_{\pi(h_{\sharp}, m_{\sharp})}, \eta_{\pi(h_{\sharp}, m_{\sharp})})} \geq \frac{9 d_*^2}{5} - 4 \|\mathcal{A}^*(e)\|_2 d_*.
\end{equation*}
\end{proof}


\section{Experiments} \label{BD:s15}

In this section, numerical simulations are used to illustrate the performance of the proposed method. 
Section~\ref{BD:s17} gives the experimental environment, synthetic problem settings, parameters and complexities.
The synthetic problems are used in Sections~\ref{BD:s18}--\ref{BD:s20}. Specifically, 
Section~\ref{BD:s18} compares the efficiency of the proposed method to the method in~\cite{LLSW2016} and an alternating minimization method; Section~\ref{BD:s19} presents an empirical estimation for the probability of successful recovery against the number of measurements; and Section~\ref{BD:s20} shows the robustness of the proposed method. In Section~\ref{BD:s23}, an image from FLAVIA dataset~\cite{WBXWCX07}, rather than a synthetic data, is used to show the performance of the proposed method in an image deblurring problem.

\subsection{Environment, Step Size, Problem Setting and Complexities} \label{BD:s17}

The codes of Algorithm~\ref{BD:a1} are written in C++ using the library ROPTLIB~\cite{HAGH2016} through its Matlab interface. All experiments are performed in Matlab R2016b on a 64bit Windows system with 3.4GHz CPU (Intel(R) Core(TM) i7-6700). The DFT is performed using the library FFTW~\cite{FFTW05} with one thread. The code is available at \url{www.math.fsu.edu/~whuang2/papers/BDSDAQM.htm}.

In signal recovery problems, including the blind deconvolution problem, theoretical results usually require the step size to be a sufficiently small constant. However, in practice, a too small step size slows down algorithm significantly while a too large step size makes algorithm fail to converge. Therefore, heuristic ideas have been used. In~\cite{CLS2016}, the step size is given by a predetermined nonincreasing sequence. In~\cite{LLSW2016}, the step size is chosen by backtracking with initial step size $1 / d$, where $d$ is the singular value in Algorithm~\ref{BD:a4}. In this paper, the step size is given by the backtracking algorithm 
with BB initial step size~\cite{BB1988}. Specifically, suppose the iterates generated by Algorithm~\ref{BD:a5} is $\{x_k\}$ and $\eta_k = R_{x_k}^{-1}(x_{k+1})$. We use the initial step size $g(s_k, y_k) / g(y_k, y_k)$, where $s_k = \mathcal{T}_{\eta_k} \eta_k$ and $y_k = \grad f(x_{k+1}) - \mathcal{T}_{\eta_k} \grad f(x_k)$ and the vector transport $\mathcal{T}$ is defined in~\eqref{BD:e19}. \whfirrev{}{Note that this initial step size may or may not be smaller than $1 / (2 a_L)$, which is assumed by Theorem~\ref{BD:th3}.}

The matrix $B$ is the first $K$ columns of a unitary $L \times L$ DFT matrix. The matrix $A$ is a Gaussian random matrix. The initial iterate, $\rho$, and $\mu$ are given using the same method as~\cite{LLSW2016}, i.e., the initial iterate is the normalized leading singular vectors of $\mathcal{A}^*(y)$, $\rho = d^2 / 100$ and $\mu = 6 \sqrt{L / (K + N)} / \log(L)$. Unless otherwise indicated, the measurements $y$ are noiseless, i.e., $y = \mathcal{A}(h_{\sharp} m_{\sharp}^*)$, where $h_{\sharp}$ and $m_{\sharp}$ are Gaussian random vectors. The stopping criterion requires $\|y - \mathcal{A}(h m^*)\|_2 / \|y\|_2 \leq 10^{-8}$.

Since $B = \sqrt{L} \mathbf{F} \mathbf{B}$ and $C = \sqrt{L} \bar{\mathbf{F}} \mathbf{C}$, the cost function $\tilde{f}$ is\footnote{$\tilde{F}$ can be done similarly and therefore its complexity is not discussed here.}
\begin{align*}
\tilde{f}(\pi(h, m)) =& \|y - L \diag\left((\mathbf{F} \mathbf{B} h) (\bar{\mathbf{F}} \mathbf{C} m)^*\right)\|_2^2  + \rho \sum_{i = 1}^L G_0\left(\frac{L |b_i^* h|^2 \|m\|_2^2}{8 d^2 \mu^2}\right).
\end{align*}
If $B$ has $K$ nonzero entries, $C$ is a random dense matrix and the order of computations of multiplications is to compute $\mathbf{B} h$ and then apply FFT to the resulting vector $\mathbf{F}(\mathbf{B} h)$ and likewise for $\bar{\mathbf{F}} (\mathbf{C} m)$, then the complexity of a function evaluation is $2 \mathrm{FFT} + 2 L N + O(L)$ flops, where a flop is a float point operation~\cite[Section 1.2.4]{GV96}. Using the same idea, the complexity of a gradient evaluation after evaluating the function value at the same point is $3 \mathrm{FFT} + 2 L N + O(L)$ flops. The complexity of the vector transport~\eqref{BD:e19} is $O(K + N)$ using the intrinsic representation (see detailed discussions about intrinsic representation in~\cite{HAG2016VT}).

%
%

\subsection{Efficiency} \label{BD:s18}

Algorithm~\ref{BD:a5} is compared to the algorithm in~\cite{LLSW2016} and an alternating minimization algorithm. The alternating minimization algorithm is stated in Algorithm~\ref{BD:a6}, which approximately optimizes over $h$ and $m$ alternatively. Specifically, when one of $h$ and $m$ is fixed, the function $F$ in~\eqref{BD:e10} is quadratic and the step size in Steps~\ref{BD:a6:st1} and~\ref{BD:a6:st2} has a closed form which can be compute cheaply.
\begin{algorithm}
\caption{An alternating minimization algorithm}
\label{BD:a6}
\begin{algorithmic}[1]
\STATE $k\leftarrow0$;
\FOR {$k = 0, 1, 2, \ldots$}
\STATE $t_* = \arg\min_{t > 0} F(h_k, m_k - t \nabla_{m_k} F(h_k, m_k))$ and set $m_{k+1} = m_k - t_* \nabla_{m_k} F(h_k, m_k)$; \label{BD:a6:st1}
\STATE $t_* = \arg\min_{t > 0} F(h_k - t \nabla_{h_k} F(h_k, m_{k+1}), m_{k+1})$ and set $h_{k+1} = h_k - t_* \nabla_{h_k} F(h_k, m_{k+1})$; \label{BD:a6:st2}
\ENDFOR
\end{algorithmic}
\end{algorithm}
Let NCBT denote the algorithm in~\cite{LLSW2016}, NCBB denote the algorithm in~\cite{LLSW2016} with modified initial step size: using BB step size rather than $1/d$; AMA denote the alternating minimization algorithm; and ROBB denote the proposed Algorithm~\ref{BD:a1} with the BB initial step size. NCBT, NCBB, and AMA are implemented in Matlab. The parameters $K$, $N$ are set to be $100$ and $100$ respectively.

Table~\ref{BD:t1} reports the results of an average of 100 random runs for $L = 400$ and $600$. Since the algorithms are performed on different languages, it is unfair to compare their computational time. Therefore, the machine-independent operations are reported. It is shown that using BB step size as the initial step size improves efficiency significantly. Increasing the number of measurements reduces the difficulty of the optimization problem in the sense that the numbers of various operations decrease. ROBB method outperforms the other methods in the sense that it needs fewer number of all the operations to reach a similar accuracy. 

\begin{table}[H]
\begin{center}
\scriptsize
\caption{An average of 100 random runs. $nFFT$ denotes the number of fast Fourier transforms (or inverst FFT). $nBh$ and $nCm$ denote the numbers of matrix vector multiplications $\mathbf{B} h$ and $\mathbf{C} m$ respectively. Note that $nBh = nCm$. $RMSE$ denotes the relative error $\frac{\|h m^T -h_{\sharp} m_{\sharp}^T\|}{\|h_{\sharp}\|\|m_{\sharp}\|}$. The subscript $k$ indicates a scale of $10^k$. The algorithms NCBT, NCBB, AMA, and ROBB are introduced in the first paragraph of Section~\ref{BD:s18}.}
\vspace{1em}
\label{BD:t1}
\begin{tabular}{c|cccc|cccc}
    \hline
    &  \multicolumn{4}{c|}{$L = 400$} &  \multicolumn{4}{c}{$L = 600$} \\
    \hline
Algorithms & NCBT& NCBB & AMA & ROBB & NCBT& NCBB & AMA & ROBB \\
\hline
$nBh$/$nCm$ & $1040$ & $349$ & $718$ & $208$ & $403$ & $162$ & $294$ & $122$ \\
\hline
$nFFT$ & $2533$ & $865$ & $1436$ & $518$ & $984$ & $401$ & $588$ & $303$ \\
\hline
$RMSE$ & $3.73_{-8}$ & $2.24_{-8}$ & $3.67_{-8}$ & $2.20_{-8}$ & $2.39_{-8}$ & $1.48_{-8}$ & $2.34_{-8}$ & $1.42_{-8}$ \\
\hline
\end{tabular}
\end{center}
\end{table}

\subsection{Number of Measurements vs Success Rate} \label{BD:s19}

Both parameters $K$ and $N$ are set to be $50$. $L/(K+L)$ takes 21 values: $\{1, 1.05, 1.1, \ldots, 1.45, 1.5, 1.6, \ldots, 2.4, 2.5\}$. We consider an algorithm to successfully recover $\pi(h_{\sharp}, m_{\sharp})$ if the RMSE of the final iterate $\pi(h, m)$ is less than $10^{-2}$, i.e., $\|h m^* - h_{\sharp} m_{\sharp}^*\|_F / \|h_{\sharp} m_{\sharp}^*\|_F \leq 10^{-2}$.

Figure~\ref{BD:f1} shows an empirical phase transition curves for the four algorithms: NCBT, NCBB, AMA, and ROBB. NCBT and NCBB perform similarly while AMA and ROBB outperform NCBT and NCBB. Given the same number of measurements, the Riemannian method ROBB has the largest successful recovery probability among the four methods.

\begin{figure}[ht]
\centering
\includegraphics[width=0.8\textwidth]{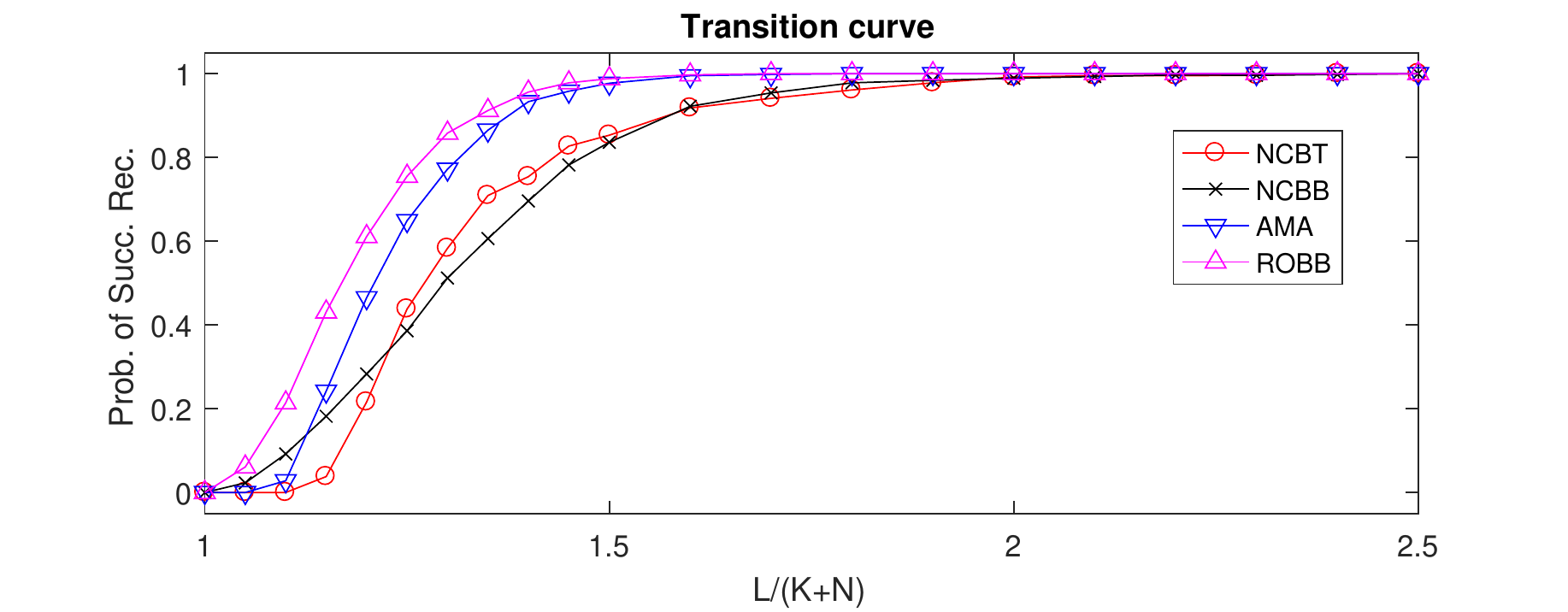}
\caption{Empirical phase transition curves for 1000 random runs.  The algorithms NCBT, NCBB, AMA, and ROBB are introduced in the first paragraph of Section~\ref{BD:s18}.}\label{BD:f1}
\end{figure}


\subsection{Noisy Measurements} \label{BD:s20}

Figure~\ref{BD:f3} shows the relationships between RMSE and SNR for Algorithm~\ref{BD:a1} with $L = 500$ and $1000$, and $K = N = 100$. The algorithm stops when the norm of the initial gradient over the norm of the last gradient is less than $10^{-12}$.
The noise measurements $y$ are given by $\mathcal{A}(h_{\sharp} m_{\sharp}^*) + e$, where the noise $e$ is $\frac{\tau w}{\|w\|_2 \|\mathcal{A}(h_{\sharp} m_{\sharp}^*)\|_2}$ for some positive value $\tau$.
Clearly, increasing the number of measurements improves the accuracy, i.e., reduces the RMSE. In addition, the curves indicate that inreasing SNR in dB reduces the RMSE in dB linearly.

\begin{figure}[ht]
\centering
\includegraphics[width=0.8\textwidth]{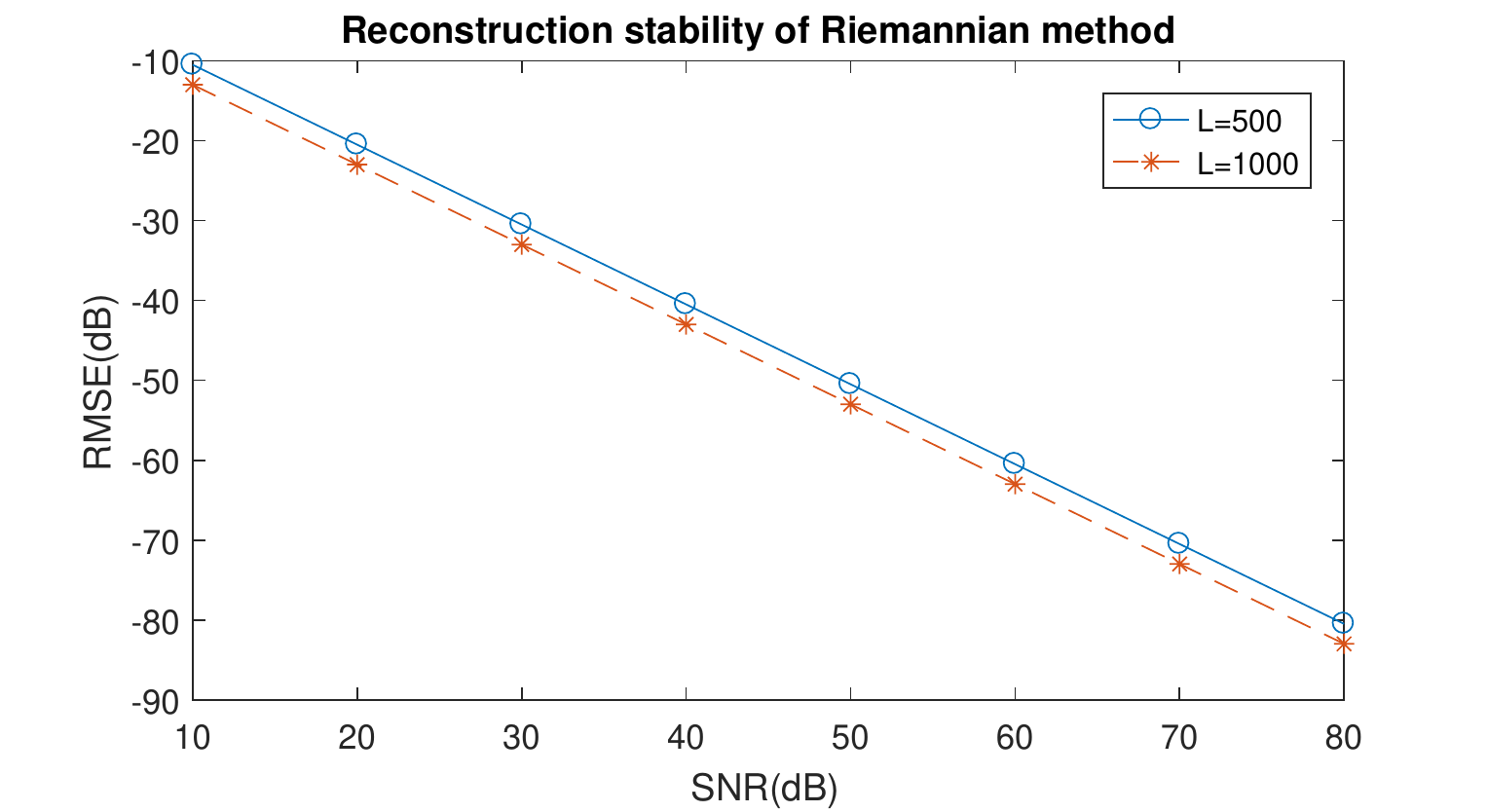}
\caption{Performance of Algorithm~\ref{BD:a5} under different SNR.}\label{BD:f3}
\end{figure}

\subsection{Natural Image} \label{BD:s23}
\whfirrev{}{
A leaf image (Figure~\ref{BD:f4}) with $1024 \times 1024$ pixels from FLAVIA dataset~\cite{WBXWCX07} is used to test the performance of Algorithm~\ref{BD:a2} on an image deblurring problem.

\begin{figure}[H]
\centering
\includegraphics[width=0.4\textwidth]{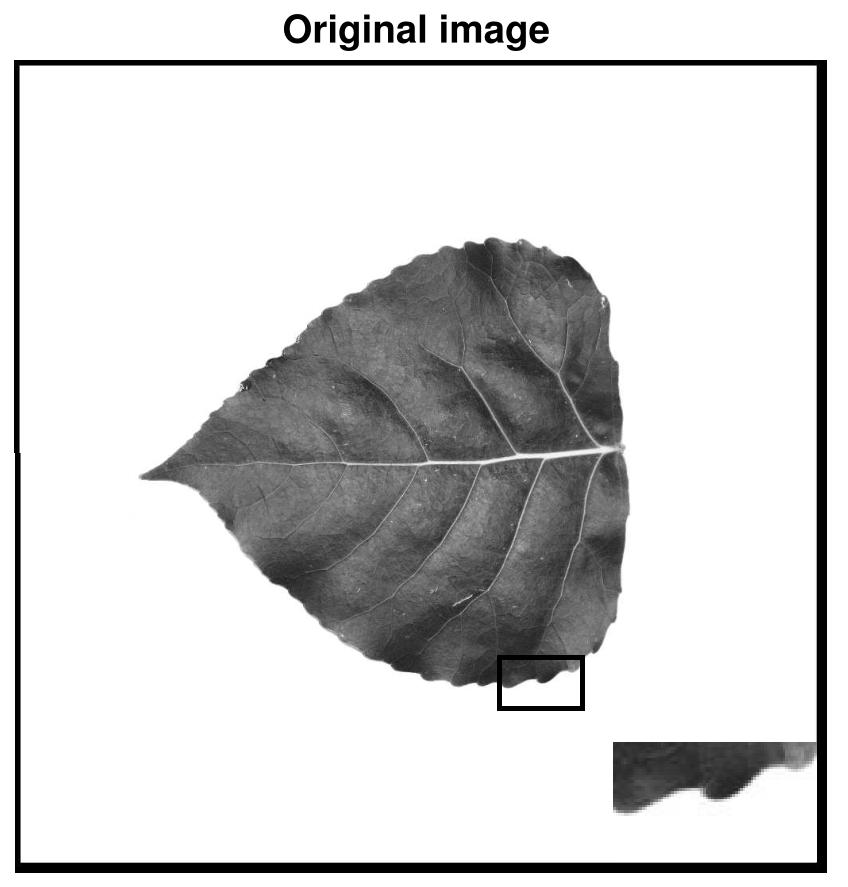}
\caption{The tested image}\label{BD:f4}
\end{figure}

A blurred image is a convolution of the original image with a blurring kernel.
Let $\mathbf{y}$, $\mathbf{\hat{y}}$ and $\kappa$ denote the vectors obtained by reshaping the blurred image, original image and the blurring kernel respectively. The measurement $y$ is therefore the vector $\mathbf{F} \mathbf{y} / \sqrt{L}$.  The matrix $\mathbf{B}$ is formed by a subset of columns of a reshaped 2D frequency Fourier matrix, where the columns correspond to the nonzero entries in $\kappa$.
Since most natural images are approximately sparse in Haar wavelet basis, the columns of $\mathbf{C} \in \mathbb{C}^{L \times N}$ are set as the $N$ most-significant columns in Haar wavelet matrix $\mathbf{W}$, where the $N$ most-significant columns denote the columns corresponding to the $N$ largest coefficients of Haar wavelet transform of the original image (the $N$ largest entries in $\mathbf{W}^T \mathbf{\hat{y}}$). However, the original image, or equivalently $\mathbf{\hat{y}}$, is unknown. Thus, we use $\mathbf{W}^T \mathbf{y}$ instead to form $\mathbf{C}$.

The number of measurements $L$ is $1024*1024 = 1048576$. The number of columns $N$ in $\mathbf{C}$ is chosen to be 20000. The number of nonzero entries in a blurring kernel is given later.

\paragraph{Motion kernel with known support:} The motion kernel\footnote{\whcomm{}{The blurring kernel is obtained by Matlab commands: ``\text{fspecial('motion', 50, 45)}''}.} and the corresponding blurred image are shown in (a) and (b) of Figure~\ref{BD:f5}. We assume that the support of the blurring kernel is known.

\begin{figure}[H]
\centering
\includegraphics[width=1\textwidth]{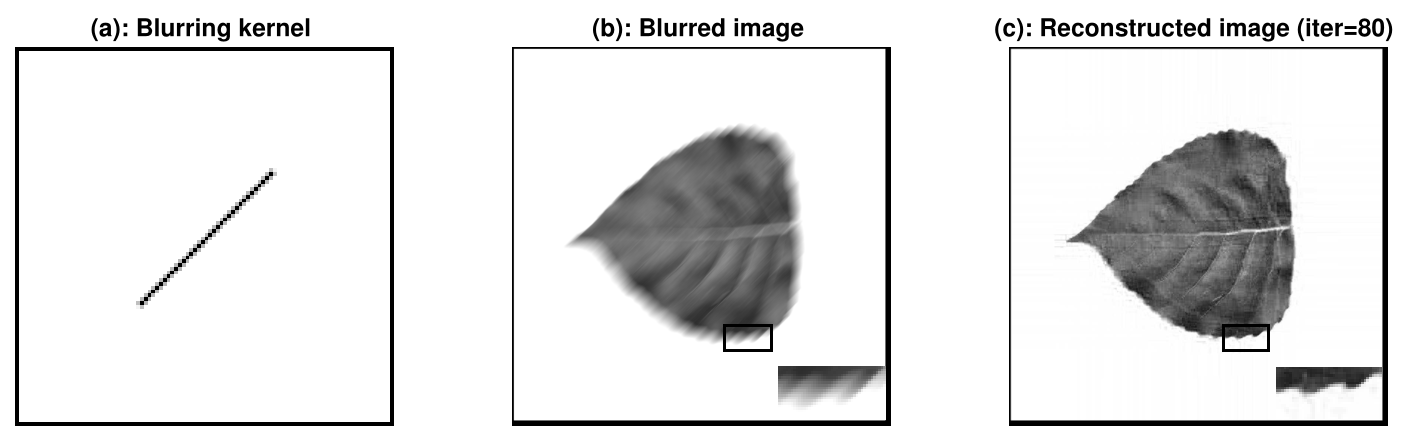}
\vspace{-2em}
\caption{Left: the blurring kernel; middle: the blurred image; right: the reconstructed image}\label{BD:f5}
\end{figure}

The number of nonzero entries in the blurring kernel is $K = 109$. The most-significant wavelet subspaces with $N = 20000$ are able to capture 97\% of the energy in the original image. Algorithm~\ref{BD:a2} stops when the number of iterations is $20, 40, 60, 80, 100, 120, 140$, $160$, or $2000$. The reconstructed image using 80 iterations is given in (c) of Figure~\ref{BD:f5}.  As shown in Table~\ref{BD:t3}, a higher accuracy does not improve the recovery performance in the sense that the relative error $relerr$, defined in Table~\ref{BD:t3}, does not necessarily decrease as the number of iterations increases. In the later experiments, the number of iterations is set to be 80.
\begin{table}[H]
\begin{center}
\scriptsize
\caption{The computational costs for multiple values of $N$ in Figure~\ref{BD:f5}. $nBh$, $nCm$, and $nFFT$ are defined in Table~\ref{BD:t1}. $t$ denotes the computational time in seconds.  $relres$ denotes $\|y - \diag(Bh m^* C^*)\|_2 / \|y\|_2$. $rellerr$ denotes $\left\|\mathbf{\hat{y}} - \frac{\|\mathbf{y}\|}{\|\mathbf{y}_f\|} \mathbf{y}_f\right\| / \|\mathbf{\hat{y}}\|$, where $\mathbf{y}_f$ denotes the vector obtained by reshaping a reconstructed image.}
\label{BD:t3}
\vspace{1em}
\begin{tabular}{c|ccccccccc}
    \hline
$num.$ $of$ $iter.$    & 20 & 40 & 60 & 80 & 100 & 120 & 140 & 160 & 2000 \\
\hline
$relres$ & $3.9_{-3}$ & $2.6_{-3}$ & $2.5_{-3}$ & $2.4_{-3}$ & $2.3_{-3}$ & $2.3_{-3}$ & $2.2_{-3}$ & $2.2_{-3}$ & $2.0_{-3}$ \\
\hline
$nBh/nCm$ & 43 & 85 & 125 & 168 & 210 & 253 & 294 & 335 & 4055\\
\hline
$nFFT$ & 107 & 211 & 311 & 417 & 521 & 627 & 729 & 831 & 10111\\
\hline
$relerr$ & $4.1_{-2}$ & $3.8_{-2}$ & $3.8_{-2}$ & $3.8_{-2}$ & $4.1_{-2}$ & $4.2_{-2}$ & $4.3_{-2}$ & $4.4_{-2}$ & $6.4_{-2}$ \\
\hline
\end{tabular}
\end{center}
\end{table}

\paragraph{Other kernels with known supports:} Two blurring kernels, their corresponding blurred images, and the reconstructed images are shown in Figures~\ref{BD:f6} and~\ref{BD:f7}. The blurring kernel in Figure~\ref{BD:f6} is from the function $\sin$ and has 153 nonzero entries. The covariance of the Gaussian kernel in Figure~\ref{BD:f7} is $V = \begin{bmatrix} 1 & 0.8 \\ 0.8 & 1 \end{bmatrix}$ and the number of nonzero entries is 181.

The number of iterations is 80 and their computational times are approximately 48 seconds. The relative errors $relerr$, defined in Table~\ref{BD:t3}, of the ``$\sin$'' kernel and the Gaussian kernel are $0.0398$ and $0.0890$, respectively. We can see that the Riemannian method is able to recover a reasonable image for a more complex kernel.

\begin{figure}[H]
\centering
\includegraphics[width=1\textwidth]{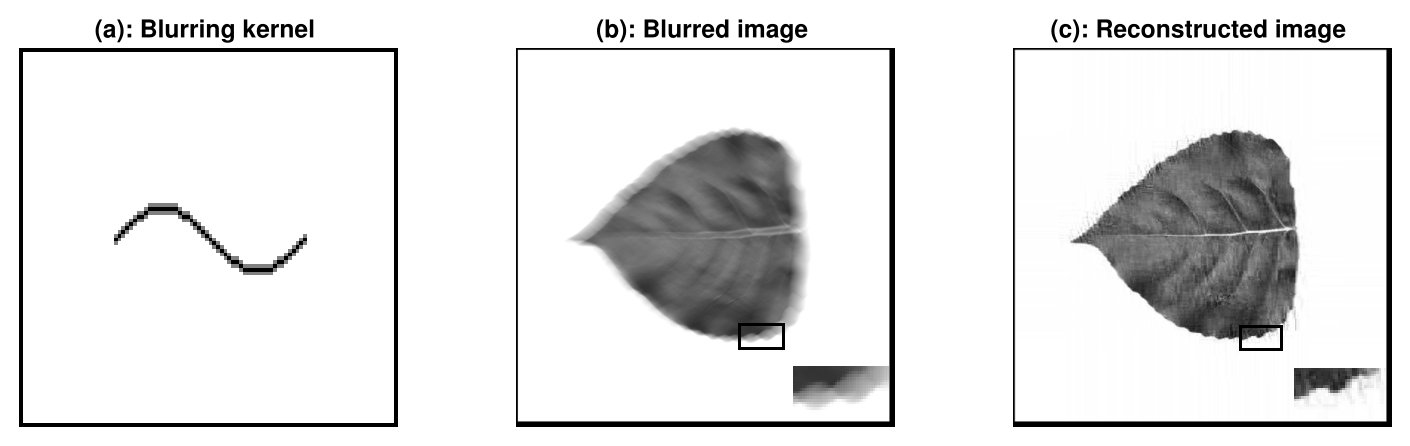}
\vspace{-2em}
\caption{Left: the blurring kernel; middle: the blurred image; right: the reconstructed image}\label{BD:f6}
\end{figure}

\begin{figure}[H]
\centering
\includegraphics[width=1\textwidth]{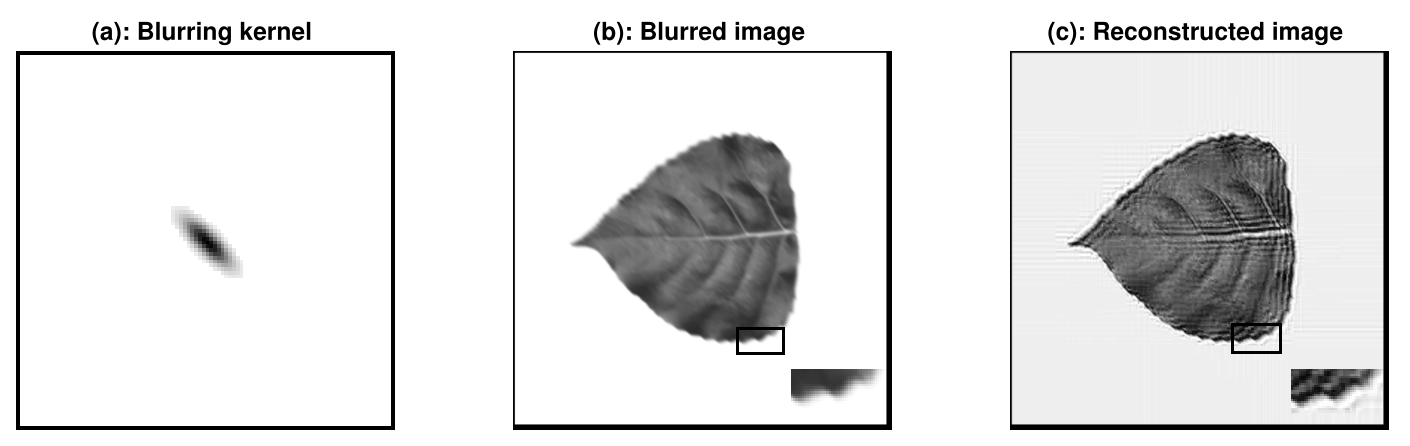}
\vspace{-2em}
\caption{Left: the blurring kernel; middle: the blurred image; right: the reconstructed image}\label{BD:f7}
\end{figure}

\paragraph{Motion kernel with unknown support:} The motion kernel in (a) of Figure~\ref{BD:f5} is used. In practice, the true support of the kernel is unknown and an estimation is usually not exact. Figure~\ref{BD:f8} shows the recovery performance of the Riemannian method when an inexact support is used. We use four inexact supports, which are obtained by enlarging the true support by 1, 2, 3, and 4 pixels. Unsurprisingly, the better the estimation of the support is, the better the reconstructed image is. The Riemannian method is able to recover the image reasonably in these tests.

\begin{figure}[H]
\centering
\includegraphics[width=1\textwidth]{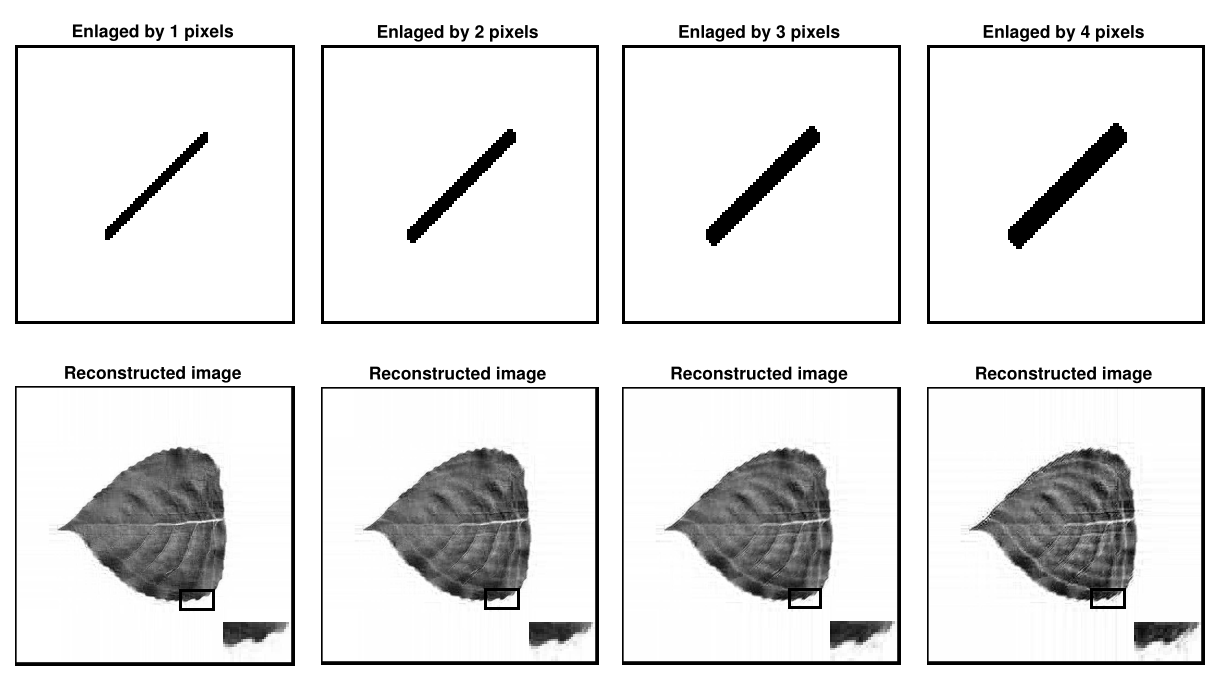}
\vspace{-2em}
\caption{Deconvolution results for unknown supports. Top: the estimated support of the blurring kernel. Bottom: the reconstructed images using the corresponding above supports. The relative error $relerr$ from left to right are $0.0440$, $0.0484$, $0.0522$, and $0.0673$, respectively.}\label{BD:f8}
\end{figure}
}
\section{Conclusion} \label{BD:s16}

In this paper, we proposed a Riemannian steepest descent method for blind deconvolution. By inspiring the proofs in~\cite{LLSW2016}, it is proven that the Riemannian method with an appropriate initialization has high probability to successfully recover the desired signal. Since a quotient manifold is considered, the penalty term that is used to control the norm of $h$ and $m$ is not necessary. The Hessian at the desired solution is proven to be well-conditioned with high probability and therefore the optimization is not difficult locally. Numerical experiments show that the Riemannian steepest descent method is robust to noise, is the most efficient algorithm, and has the empirical largest successful recovery probability among the alternating minimization algorithm and the algorithm in~\cite{LLSW2016}. 

%
%
%
%
%

\bibliographystyle{alpha}

\bibliography{WHlibrary}

\newcommand{\etalchar}[1]{$^{#1}$}
\begin{thebibliography}{WBX{\etalchar{+}}07}

\bibitem[AMS08]{AMS2008}
P.-A. Absil, R.~Mahony, and R.~Sepulchre.
\newblock {\em Optimization algorithms on matrix manifolds}.
\newblock Princeton University Press, Princeton, NJ, 2008.

\bibitem[ARR14]{ARR2014}
A.~Ahmed, B.~Recht, and J.~Romberg.
\newblock Blind deconvolution using convex programming.
\newblock {\em IEEE Transactions on Information Theory}, 60(3):1711--1732,
  March 2014.

\bibitem[BA11]{BouAbs2011}
N.~Boumal and P.-A. Absil.
\newblock {RTRMC}: A {R}iemannian trust-region method for low-rank matrix
  completion.
\newblock {\em Advances in Neural Information Processing Systems 24 (NIPS)},
  pages 406--414, 2011.

\bibitem[Bak08]{BAKER08}
C.~G. Baker.
\newblock {\em {R}iemannian manifold trust-region methods with applications to
  eigenproblems}.
\newblock PhD thesis, Florida State University, Department of Computational
  Science, 2008.

\bibitem[BB88]{BB1988}
J.~Barzilai and J.~M. Borwein.
\newblock {Two-Point Step Size Gradient Methods}.
\newblock {\em IMA Journal of Numerical Analysis}, 8:141--148, 1988.

\bibitem[BM03]{BurMon03}
S.~Burer and R.~D.~C. Monteiro.
\newblock A nonlinear programming algorithm for solving semidefinite programs
  via low-rank factorization.
\newblock {\em Mathematical Programming}, 95(2):329--357, February 2003.
\newblock doi:10.1007/s10107-002-0352-8.

\bibitem[BSA13]{Boumal2013}
N.~Boumal, A.~Singer, and P.-A. Absil.
\newblock Robust estimation of rotations from relative measurements by maximum
  likelihood.
\newblock {\em 52nd IEEE Conference on Decision and Control}, (3):1156--1161,
  December 2013.
\newblock doi:10.1109/CDC.2013.6760038.

\bibitem[CE07]{CamEgi2007}
P.~Campisi and K.~Egiazarian.
\newblock {\em {Blind image deconvolution : theory and applications}}.
\newblock CRC press, 2007.

\bibitem[CLS16]{CLS2016}
E.~J. Cand\'es, X.~Li, and M.~Soltanolkotabi.
\newblock {Phase retrieval via Wirtinger flow: theory and algorithms}.
\newblock {\em IEEE Transactions on Information Theory}, 64(4):1985--2007,
  2016.

\bibitem[FJ05]{FFTW05}
Matteo Frigo and Steven~G. Johnson.
\newblock The design and implementation of {FFTW3}.
\newblock {\em Proceedings of the IEEE}, 93(2):216--231, 2005.
\newblock Special issue on ``Program Generation, Optimization, and Platform
  Adaptation''.

\bibitem[GV96]{GV96}
G.~H. Golub and C.~F. {Van Loan}.
\newblock {\em Matrix computations}.
\newblock Johns Hopkins Studies in the Mathematical Sciences. Johns Hopkins
  University Press, third edition, 1996.

\bibitem[HAG15]{HAG13}
W.~Huang, P.-A. Absil, and K.~A. Gallivan.
\newblock A {R}iemannian symmetric rank-one trust-region method.
\newblock {\em Mathematical Programming}, 150(2):179--216, February 2015.

\bibitem[HAG16]{HAG2016VT}
W.~Huang, P.-A. Absil, and K.~A. Gallivan.
\newblock Intrinsic representation of tangent vectors and vector transport on
  matrix manifolds.
\newblock {\em Numerische Mathematik}, 136(2):523--543, 2016.

\bibitem[HAG18]{HuaAbsGal2018}
Wen Huang, P.-A. Absil, and K.~A. Gallivan.
\newblock {A Riemannian BFGS Method without Differentiated Retraction for
  Nonconvex Optimization Problems}.
\newblock {\em SIAM Journal on Optimization}, 28(1):470--495, 2018.

\bibitem[HAGH16]{HAGH2016}
Wen Huang, P.-A. Absil, K.~A. Gallivan, and Paul Hand.
\newblock {ROPTLIB}: an object-oriented {C}++ library for optimization on
  {R}iemannian manifolds.
\newblock Technical Report FSU16-14, Florida State University, 2016.

\bibitem[HGA15]{HGA2014}
W.~Huang, K.~A. Gallivan, and P.-A. Absil.
\newblock {A Broyden Class of Quasi-Newton Methods for Riemannian
  Optimization}.
\newblock {\em SIAM Journal on Optimization}, 25(3):1660--1685, 2015.

\bibitem[HGZ17]{HGZ2017}
Wen Huang, K.~A. Gallivan, and Xiangxiong Zhang.
\newblock {Solving PhaseLift by low rank Riemannian optimization methods for
  complex semidefinite constraints}.
\newblock {\em SIAM Journal on Scientific Computing}, 39(5):B840--B859, 2017.

\bibitem[Hua13]{HUANG2013}
W.~Huang.
\newblock {\em Optimization algorithms on {R}iemannian manifolds with
  applications}.
\newblock PhD thesis, Florida State University, Department of Mathematics,
  2013.

\bibitem[JC93]{JC1993}
S.~M. Jefferies and J.~C. Christou.
\newblock {Restoration of astronomical images by iterative blind
  deconvolution}.
\newblock {\em The Astrophysical Journal}, page 415:862, 1993.

\bibitem[LLJB15]{LLJB2015}
K.~Lee, Y.~Li, M.~Junge, and Y.~Bresler.
\newblock {Blind Recovery of Sparse Signals from Subsampled Convolution}.
\newblock pages 1--43, 2015.

\bibitem[LLSW16]{LLSW2016}
Xiaodong Li, Shuyang Ling, Thomas Strohmer, and Ke~Wei.
\newblock Rapid, robust, and reliable blind deconvolution via nonconvex
  optimization.
\newblock {\em CoRR}, abs/1606.04933, 2016.

\bibitem[LS15]{LingStro2015}
S.~Ling and T.~Strohmer.
\newblock {Self-Calibration and Biconvex Compressive Sensing}.
\newblock {\em Inverse Problems}, 31(11):115002, 2015.

\bibitem[LWDF11]{LWDF2011}
A.~Levin, Y.~Weiss, F.~Durand, and W.~T. Freeman.
\newblock {Understanding and evaluating blind deconvolution algorithms}.
\newblock {\em IEEE Transactions on Pattern Analysis and Machine Intelligence},
  33(12):2354--2367, 2011.

\bibitem[Mis14]{Mishra2014}
B.~Mishra.
\newblock {\em {A Riemannian approach to large-scale constrained least-squares
  with symmetries}}.
\newblock PhD thesis, University of Liege, 2014.

\bibitem[NW06]{NocWri2006}
J.~Nocedal and S.~J. Wright.
\newblock {\em {Numerical Optimization}}.
\newblock Springer, second edition, 2006.

\bibitem[RW12]{RinWir2012}
W.~Ring and B.~Wirth.
\newblock Optimization methods on {R}iemannian manifolds and their application
  to shape space.
\newblock {\em SIAM Journal on Optimization}, 22(2):596--627, January 2012.
\newblock doi:10.1137/11082885X.

\bibitem[Sat16]{Sato2015}
H.~Sato.
\newblock A {D}ai--{Y}uan-type {R}iemannian conjugate gradient method with the
  weak {W}olfe conditions.
\newblock {\em Computational Optimization and Applications}, 64(1):101--118,
  May 2016.

\bibitem[SI15]{SI2015}
H.~Sato and T.~Iwai.
\newblock {A new, globally convergent Riemannian conjugate gradient method}.
\newblock {\em Optimization}, 64(4):1011--1031, February 2015.

\bibitem[SL16]{SL2016}
Ruoyu Sun and Zhi~Quan Luo.
\newblock {Guaranteed Matrix Completion via Non-Convex Factorization}.
\newblock {\em IEEE Transactions on Information Theory}, 62(11):6535--6579,
  2016.

\bibitem[Van13]{Vandereycken2013}
B.~Vandereycken.
\newblock Low-rank matrix completion by {R}iemannian optimization---extended
  version.
\newblock {\em SIAM Journal on Optimization}, 23(2):1214--1236, 2013.

\bibitem[WBSJ15]{WBSJ2015}
G.~Wunder, H.~Boche, T.~Strohmer, and P.~Jung.
\newblock {Sparse Signal Processing Concepts for Efficient 5G System Design}.
\newblock {\em IEEE Access}, 3:195--208, 2015.

\bibitem[WBX{\etalchar{+}}07]{WBXWCX07}
S.~G. Wu, F.~S. Bao, E.~Y. Xu, Y.-X. Wang, Y.-F. Chang, and Q.-L. Xiang.
\newblock A leaf recognition algorithm for plant classification using
  probabilistic neural network.
\newblock {\em 2007 IEEE International Symposium on Signal Processing and
  Information Technology}, pages 11--16, 2007.
\newblock arXiv:0707.4289v1.

\bibitem[WCCL16]{WCCL2016}
K.~Wei, J.-F. Cai, T.~F. Chan, and S.~Leung.
\newblock {Guarantees of Riemannian Optimization for Low Rank Matrix
  Completion}.
\newblock (1), 2016.

\bibitem[WP98]{WP1998}
X.~Wang and H.~V. Poor.
\newblock {Blind Equalization and Multiuser Detection in Dispersive CDMA
  Channels}.
\newblock {\em IEEE Transactions on Communications}, 46(1):91--103, 1998.

\bibitem[YK94]{You1994}
T.-L. You and M.~Kaveh.
\newblock {A simple algorithm for joint blur identification and image
  restoration}.
\newblock {\em Proceedings - International Conference on Image Processing,
  ICIP}, 3(3):167--171, 1994.

\bibitem[Zhu17]{Zhu2016}
X.~Zhu.
\newblock A {R}iemannian conjugate gradient method for optimization on the
  {S}tiefel manifold.
\newblock {\em Computational Optimization and Applications}, 67(1):73--110,
  2017.

\end{thebibliography}

\appendix

\section{Four Conditions and the Proof of Theorem~\ref{BD:th3}} \label{BD:s8}

The convergence analysis of the Riemannian steepest descent method follows the spirit of the analysis in~\cite{LLSW2016} and the analyses both rely on the four conditions: local RIP condition, robustness condition, local regularity condition and local smoothness condition. There exist differences and the main ones are highlighted as follows. The differences ease the proofs of the Riemannian method in general.
\begin{itemize}
\item The cost function $\tilde{f}$ does not include the penalty terms for the norm of $h$ and $m$. Therefore, the penalty terms are not considered in the convergence analysis.
\item Since any representation in $\pi^{-1}(\pi(h, m))$ can be used and this does not influence the sequence $\{\pi(h_k, m_k)\}$ of $\mathcal{Q}_1$ generated by Algorithm~\ref{BD:a1}, we can always assume $\|h_k\|_2 = \|m_k\|_2$ without loss of generality.
\item The Riemannian gradient is different from the Wirtinger derivative. For the cost function $F$ in~\eqref{BD:e10}, if $\|h\|_2 = \|m\|_2$, then by~\eqref{BD:e28} and~\eqref{BD:e27}, we have
\begin{equation*}
\grad F(h, m) = \frac{2}{\|h m^*\|_2} \nabla^w F(h, m).
\end{equation*}
\end{itemize}
Let $\pi(h_{\sharp}, m_{\sharp})$ denote the ground truth and $\Upsilon_{\tilde{f}}$ denote $\{\pi(h, m) \mid \tilde{f}(\pi(h, m)) \leq \frac{1}{3} \varepsilon^2 d_*^2 + \|e\|_2^2\}$.

\begin{cond}[Local RIP condition] \label{BD:c1}
The following local Restricted Isometry Property for $\mathcal{A}$ holds uniformly for all $(h, m) \in \Pi_{\varepsilon}$:
\begin{equation*}
\frac{3}{4} \|h m^* - h_{\sharp} m_{\sharp}^*\|_F^2 \leq \|\mathcal{A}(h m^* - h_{\sharp} m_{\sharp}^*)\|_2^2 \leq \frac{5}{4} \|h m^* - h_{\sharp} m_{\sharp}^*\|_F^2.
\end{equation*}
\end{cond}

\begin{cond} [Robustness condition] \label{BD:c2}
If $L \geq C_\gamma (\sigma^2 / \varepsilon^2 + \sigma / \varepsilon) \max(K, N) \log(L)$, then with high probability, it holds that
\begin{equation*}
\|\mathcal{A}^*(e)\|_2 \leq \frac{\varepsilon d_*}{10 \sqrt{2}}.
\end{equation*}
\end{cond}

\begin{cond} [Local regularity condition] \label{BD:c3}
There exists a regularity constant $a_0 > 0$ such that
\begin{equation*}
\|\grad \tilde{f}(\pi(h, m))\|_{g_{\pi(h, m)}}^2 \geq a_0 [\tilde{f}(\pi(h, m)) - c]_+
\end{equation*}
for all $\pi(h, m) \in \Omega_{\mu} \cap \Pi_{\varepsilon}$, where $c = \|e\|_2^2 + 1700 \|\mathcal{A}^*(e)\|_2^2$ and $a_0 = 1 / 1500$.
\end{cond}

\begin{cond} [Local smoothness condition] \label{BD:c4}
Define the lifting function $\hat{f}_{\pi(h, m)}: \T_{\pi(h, m)} \mathcal{Q}_1 \rightarrow \mathbb{R}: \eta_{\pi(h, m)} = \tilde{f} \circ \tilde{R}_{\pi(h, m)} (\eta_{\pi(h, m)})$. There exists a constant $a_L$ such that
\begin{equation*}
\|\grad \hat{f}_{\pi(h, m)}(t \eta_{\pi(h, m)}) - \grad \hat{f}_{\pi(h, m)}(0)\|_{g_{(h, m)}} \leq a_L t \|\eta_{\pi(h, m)}\|_{g_{(h, m)}}, \;\;\; \forall 0 \leq t \leq 1,
\end{equation*}
for all $\pi(h, m)$ and $\eta_{\pi(h, m)} \in \T_{\pi(h, m)} \mathcal{Q}_1$ such that $\tilde{R}_{\pi(h, m)}(t \eta_{\pi(h, m)}) \in \Upsilon_{\tilde{f}} \cap \Pi_\varepsilon, \forall 0 \leq t  \leq 1$, where $\tilde{R}$ is defined in~\eqref{BD:e34}.
\end{cond}
Conditions~\ref{BD:c1} and~\ref{BD:c2} are the same as~\cite[Conditions~5.1 and~5.2]{LLSW2016} and have been proven therein. Conditions~\ref{BD:c3} and~\ref{BD:c4} are different from\cite[Conditions~5.3 and~5.4]{LLSW2016} since we use different penalty term and different gradient and norm. The proofs of Conditions~\ref{BD:c3} and~\ref{BD:c4} are given in Section~\ref{BD:s7}.

Lemma~\ref{BD:le10} generalizes~\cite[Lemma 6.1]{LLSW2016} and is used in Lemma~\ref{BD:le9}.
\begin{lemm} [Riemannian descent lemma] \label{BD:le10}
Suppose Condition~\ref{BD:c4} holds. Then
\begin{equation*}
\hat{f}_{\pi(h, m)}(\eta_{\pi(h, m)}) \leq \hat{f} (0_{\pi(h, m)}) + g_{\pi(h, m)} (\eta_{\pi(h, m)}, \grad \tilde{f} (\pi(h, m))) + a_L \|\eta_{\pi(h, m)}\|_g^2,
\end{equation*}
where $0_{\pi(h, m)}$ denotes the origin of $\T_{\pi(h, m)} \mathcal{Q}_1$.
\end{lemm}
\begin{proof}
Define $\beta(t) = \hat{f}_{\pi(h, m)}(t \eta_{\pi(h, m)})$. We have
$\frac{d}{d t} \beta(t) = g \left(\grad \hat{f}_{\pi(h, m)}(t \eta_{\pi(h, m)}), \eta_{\pi(h, m)}\right).$
It follows from the Fundamental Theorem of Calculus that
\begin{align*}
\hat{f}_{\pi(h, m)}&(\eta_{\pi(h, m)}) - \hat{f} (0) = \int_0^1 \frac{d}{d t} \beta(t) d t = \int_0^1 g \left(\grad \hat{f}_{\pi(h, m)}(t \eta_{\pi(h, m)}), \eta_{\pi(h, m)}\right) d t \\
= & g \left(\grad \hat{f}_{\pi(h, m)}(0_{\pi(h, m)}), \eta_{\pi(h, m)}\right) + \int_0^1 g \left(\grad \hat{f}_{\pi(h, m)}(t \eta_{\pi(h, m)}) - \grad \hat{f}_{\pi(h, m)}(0_{\pi(h, m)}), \eta_{\pi(h, m)}\right) d t \\
\leq & g \left(\grad \hat{f}_{\pi(h, m)}(0_{\pi(h, m)}), \eta_{\pi(h, m)}\right) + a_L \|\eta_{\pi(h, m)}\|_g^2.
\end{align*}
Therefore, the result holds since $\grad \hat{f}_{\pi(h, m)}(0_{\pi(h, m)}) = \grad \tilde{f}(\pi(h, m))$, see~\cite[(4.4)]{AMS2008}.
\end{proof}

Lemma~\ref{BD:le9} is used in the proof of Theorem~\ref{BD:th3}. Note that this lemma follows from~\cite[Section 5.1]{LLSW2016}.
\begin{lemm} \label{BD:le9}
The following properties hold under some of the four conditions.
\begin{enumerate}
\item \label{BD:p1} Under conditions~\ref{BD:c1} and~\ref{BD:c2}, function $f$ in~\eqref{BD:e14} satisfies
\begin{equation} \label{BD:e15}
\|e\|_2^2 + \frac{3}{4} \Delta^2 - \frac{\varepsilon \Delta d_*}{5} \leq f(\pi(h, m)) \leq \|e\|_2^2 + \frac{5}{4} \Delta^2 + \frac{\varepsilon \Delta d_*}{5}
\end{equation}
for all $\pi(h, m) \in \Omega_{\mu} \cap \Pi_{\varepsilon}$, where $\Delta = \|h m^* - h_{\sharp} m_{\sharp}^*\|_F$.
\item \label{BD:p2}  It holds that $\Upsilon_{\tilde{f}} \subset \Omega_\mu$; under conditions~\ref{BD:c1} and~\ref{BD:c2}, we have $\Upsilon_{\tilde{f}} \cap \Pi_\varepsilon \subset \Pi_{\frac{9}{10} \varepsilon}$.
\item \label{BD:p3} Under conditions~\ref{BD:c1} and~\ref{BD:c2}, if $\pi(h_1, m_1) \in \Pi_\varepsilon$ and $\pi\left((1 - \lambda) h_1 + \lambda h_2, (1 - \lambda) m_1 + \lambda m_2\right) \in \Upsilon_{\tilde{f}}$ for all $\lambda \in [0, 1]$, then $\pi(h_2, m_2) \in \Pi_{\varepsilon}$.
\item \label{BD:p4}  Under conditions~\ref{BD:c1}, \ref{BD:c2} and~\ref{BD:c4}, suppose the step size $\alpha \leq \frac{1}{2 a_L}$, where $a_L$ is defined in Condition~\ref{BD:c4} and $\pi(h_k, m_k) \in \Pi_{\varepsilon} \cap \Upsilon_{\tilde{f}}$, then it holds that $\pi(h_{k + 1}, m_{k + 1}) \in \Pi_{\varepsilon} \cap \Upsilon_{\tilde{f}}$ and
    \begin{equation*}
    \tilde{f}(\pi(h_{k+1}, m_{k+1})) \leq \tilde{f}(\pi(h_k, m_k)) - \frac{\alpha}{2} \|\grad \tilde{f}(\pi(h_k, m_k))\|_{g}^2.
    \end{equation*}
\end{enumerate}
\end{lemm}
\begin{proof}
\eqref{BD:p1}: This has been proven in~\cite[(5.5), (5.6)]{LLSW2016}.

\eqref{BD:p2}: If $\pi(h, m) \notin \Omega_{\mu}$, then $G(\pi(h, m)) \geq \rho G_0(\frac{2d_*^2}{d^2})$. It follows that
\begin{align*}
\tilde{f}(\pi(h, m)) \geq& \rho G_0(\frac{2d_*^2}{d^2}) \geq (d^2 + 2.5 \|e\|_2^2) \left(\frac{2 d_*^2}{d^2} - 1\right)^2 \\
\geq& d^2 \left(\frac{2 d_*^2}{d^2} - 1\right)^2 + 2.5 \|e\|_2^2 \left(\frac{2 d_*^2}{d^2} - 1\right)^2 \geq \frac{1}{3} \varepsilon^2 d_*^2 + \|e\|_2^2,
\end{align*}
where $\rho \geq d^2 + 2.5 \|e\|_2^2$, $\varepsilon < 1/15$, and $0.9 d_* \leq d \leq 1.1 d_*$. Therefore, $\pi(h, m) \notin \Upsilon_{\tilde{f}}$ and $\Upsilon_{\tilde{f}} \subset \Omega_{\mu}$.

For any $\pi(h, m) \in \Upsilon_{\tilde{f}} \cap \Pi_\varepsilon$, it holds that $\pi(h, m) \in \Omega_{\mu} \cap \Pi_\varepsilon$. It follows from~\eqref{BD:e15} that $\|e\|_2^2 + \frac{3}{4} \Delta^2 - \frac{\varepsilon \Delta d_*}{5} \leq f(\pi(h, m)) \leq \tilde{f}(\pi(h, m)) \leq \frac{1}{3} \varepsilon^2 d_*^2 + \|e\|_2^2$, which can be simplified into $45 \Delta^2 - 12 \varepsilon d_* \Delta - 20 \varepsilon^2 d_*^2 \leq 0$. This implies $\Delta \leq 0.9 \varepsilon d_*$ and hence $\Upsilon_{\tilde{f}} \cap \Pi_\varepsilon \subset \Pi_{\frac{9}{10} \varepsilon}$.

\eqref{BD:p3}: Prove by contradiction. If $\pi(h_2, m_2) \notin \Pi_\varepsilon$, then there exists
$\lambda_0$ such that
$$
\|\left((1 - \lambda) h_1 + \lambda h_2\right) \left((1 - \lambda) m_1 + \lambda m_2\right)^* - h_1 m_1^*\| = \varepsilon d_*,
$$
i.e., $\pi\left( (1 - \lambda) h_1 + \lambda h_2, (1 - \lambda) m_1 + \lambda m_2 \right) \in \partial \Pi_{\varepsilon}$. Therefore, it follows from
$$
\pi\left( (1 - \lambda) h_1 + \lambda h_2, (1 - \lambda) m_1 + \lambda m_2 \right) \in \Upsilon_{\tilde{f}}
$$ that $\pi\left( (1 - \lambda) h_1 + \lambda h_2, (1 - \lambda) m_1 + \lambda m_2 \right) \in \Pi_{0.9 \varepsilon}$. This is a contradiction.

\eqref{BD:p4}: If $\grad \tilde{f}(\pi(h_k, m_k)) = 0$, then $\pi(h_{k+1}, m_{k+1}) = \pi(h_k, m_k) \in \Pi_{\varepsilon} \cap \Upsilon_{\tilde{f}}$. Suppose $\grad \tilde{f}(\pi(h_k, m_k)) \neq 0$. Define the function
\begin{equation*}
\phi(\lambda) = \hat{f}_{\pi(h_k, m_k)}(- \lambda \grad \tilde{f}(\pi(h_k, m_k))).
\end{equation*}
It follows that $\phi'(0) = - \|\grad \tilde{f}(\pi(h_k, m_k))\|_g^2 < 0$. By Lemma~\ref{BD:le10}, we have
\begin{align}
\phi(\lambda) \leq& \phi(0) - \lambda \|\grad \tilde{f}(\pi(h_k, m_k))\|_g^2 + a_L \lambda^2 \|\grad \tilde{f}(\pi(h_k, m_k))\|_g^2 \\
=& \phi(0) + (a_L \lambda^2 - \lambda) \|\grad \tilde{f}(\pi(h_k, m_k))\|_g^2 \leq \phi(0) \label{BD:e26}
\end{align}
for all $\lambda \in [0, \alpha]$. Therefore, by~\eqref{BD:p3} in Lemma~\ref{BD:le9}, we have $\pi(h_{k+1}, m_{k+1}) \in \Pi_{\varepsilon} \cap \Upsilon_{\tilde{f}}$. Using~\eqref{BD:e26} yields
\begin{equation*}
\tilde{f}(\pi(h_{k+1}, m_{k+1})) \leq \tilde{f}(\pi(h_k, m_k)) - \frac{\alpha}{2} \|\grad \tilde{f}(\pi(h_k, m_k))\|_{g}^2.
\end{equation*}
\end{proof}

Now, we are ready to prove Theorem~\ref{BD:th3}. With the four conditions, the proof follows from the proof in~\cite{LLSW2016} and we give here for completeness.
\begin{proof}
Since $\pi(h_0, m_0) \in \Omega_{\frac{1}{2}\mu} \cap \Pi_{\frac{2}{5} \varepsilon}$, we have
$$
\frac{L |b_i^* h_0|^2 \|m_0\|_2^2}{8 d^2 \mu^2} \leq \frac{L}{8 d^2 \mu^2} \frac{4 d_*^2 \mu^2}{L} \leq \frac{d_*^2}{2 d^2} < 1,
$$
where $\sqrt{L} \|B h_0\|_{\infty} \|m\|_2 \leq 2 d_* \mu$. Therefore, the penalty term $G(\pi(h_0, m_0)) = 0$. Combining~\eqref{BD:e15} and~$\Delta \leq \frac{2}{5} \varepsilon d_*$ yields
\begin{equation*}
\tilde{f}(\pi(h, m)) = f(\pi(h, m)) \leq \|e\|_2^2 + \frac{5}{4} \Delta^2 + \frac{\varepsilon \Delta d_*}{5} \leq \frac{1}{3} \varepsilon^2 d_*^2 + \|e\|_2^2,
\end{equation*}
which implies $\pi(h_0, m_0) \in \Upsilon_{\tilde{f}}$. It follows from~\eqref{BD:p4} in Lemma~\ref{BD:le9} and Condition~\ref{BD:c3} that
\begin{equation*}
\tilde{f}(\pi(h_{k + 1}, m_{k + 1})) \leq \tilde{f}(\pi(h_{k}, m_{k})) - \frac{\alpha a_0}{2} \left[ \tilde{f}(\pi(h_{k}, m_{k})) - c \right]_+
\end{equation*}
and $\pi(h_k, m_k) \in \Omega_\mu \cap \Pi_\varepsilon$ for all $k$. Therefore, we have
\begin{equation*}
\tilde{f}(\pi(h_{k + 1}, m_{k + 1})) - c \leq \left(1 - \frac{\alpha a_0}{2}\right) \left[ \tilde{f}(\pi(h_{k}, m_{k})) - c \right]_+
\end{equation*}
which implies
\begin{equation*}
\left[\tilde{f}(\pi(h_{k + 1}, m_{k + 1})) - c\right]_+ \leq \left(1 - \frac{\alpha a_0}{2}\right) \left[ \tilde{f}(\pi(h_{k}, m_{k})) - c \right]_+.
\end{equation*}
Thus, we obtain
\begin{equation*}
\left[\tilde{f}(\pi(h_{k}, m_{k})) - c\right]_+ \leq \left(1 - \frac{\alpha a_0}{2}\right)^{k} \left[\tilde{f}(\pi(h_{0}, m_{0})) - c\right]_+ \leq \frac{1}{3} \left(1 - \frac{\alpha a_0}{2}\right)^{k} \varepsilon^2 d_*^2,
\end{equation*}
where we use $\tilde{f}(\pi(h_0, m_0)) \leq \frac{1}{3} \varepsilon^2 d_*^2 + \|e\|_2^2$ and $c = \|e\|_2^2 + 1700 \|\mathcal{A}^*(e)\|_2^2 \geq \|e\|_2^2$. We also have
\begin{align*}
\tilde{f}(\pi(h_k, m_k)) - \|e\|_2^2 \geq& \|\mathcal{A}(h_k m_k^* - h_{\sharp} m_{\sharp}^*)\|_2^2 - 2 \mathrm{Re}(\inner[2]{\mathcal{A}^*(e)}{h_k m_k^* - h_{\sharp} m_{\sharp}^*}) \\
\geq& \frac{3}{4} \|h_k m_k^* - h_{\sharp} m_{\sharp}^*\|_F^2 - 2 \sqrt{2} \|\mathcal{A}^*(e)\|_2 \|h_k m_k^* - h_{\sharp} m_{\sharp}^*\|_F.
\end{align*}
It follows that
\begin{equation*}
\frac{1}{3} \left(1 - \frac{\alpha a_0}{2}\right)^{k} \varepsilon^2 d_*^2 \geq \left[\tilde{f}(\pi(h_{k}, m_{k})) - c\right]_+ \geq \frac{3}{4} \|h_k m_k^* - h_{\sharp} m_{\sharp}^*\|_F^2 - 2 \sqrt{2} \|\mathcal{A}^*(e)\|_2 \|h_k m_k^* - h_{\sharp} m_{\sharp}^*\|_F - 1700 \|\mathcal{A}^*(e)\|_2^2,
\end{equation*}
which is equivalent to
\begin{equation*}
\left| \|h_k m_k^* - h_{\sharp} m_{\sharp}^*\|_F - \frac{4 \sqrt{2}}{3} \|\mathcal{A}^*(e)\|_2 \right|^2 \leq \frac{4}{9} \left(1 - \frac{\alpha a_0}{2}\right)^t \varepsilon^2 d_*^2 + (\frac{6800}{3} + \frac{32}{9}) \|\mathcal{A}^*(e)\|_2^2.
\end{equation*}
Solving for $\|h_k m_k^T - h_{\sharp} m_{\sharp}^T\|_F$ yields
\begin{equation*}
\|h_k m_k^T - h_{\sharp} m_{\sharp}^T\|_F \leq \frac{2}{3} \left(1 - \frac{\alpha a_0}{2}\right)^{k/2} \varepsilon d_* + 50 \|\mathcal{A}^*(e)\|_2.
\end{equation*}
The upper bound for $\|\mathcal{A}^*(e)\|_2$ has been proven in~\cite{LLSW2016}.
\end{proof}

\section{Proofs of Conditions~\ref{BD:c3} and~\ref{BD:c4} and Theorem~\ref{BD:th2}} \label{BD:s7}

Define function $\tilde{f}_{\mathrm{T}} = \tilde{f} \circ \pi$, $f_{\mathrm{T}} = f \circ \pi$, and $G_{\mathrm{T}} = G \circ \pi$, which is defined in the total space $\mathbb{R}_*^K \times \mathbb{R}_*^N$. Since the function value $f_{\mathrm{T}}(h, m)$ and $\|\grad f_{\mathrm{T}}(h, m)\|_g$ are independent of representation in $\pi^{-1}(\pi(h, m))$, we can always choose $h$ and $m$ such that $\|h\|_2 = \|m\|_2$. For all $\pi(h, m) \in \Pi_\varepsilon$, we have $\|h\|_2 \|m\|_2 = \|h m^*\|_F \leq \|h m^* - h_{\sharp} m_{\sharp}^*\|_F + \|h_{\sharp} m_{\sharp}^*\|_F \leq (1 + \varepsilon) d_*$ and $\|h\|_2 \|m\|_2 = \|h m^*\|_F \geq \|h_{\sharp} m_{\sharp}^*\|_F  - \|h m^* - h_{\sharp} m_{\sharp}^*\|_F \geq (1 - \varepsilon) d_*$. It follows that
\begin{equation} \label{BD:e16}
(h, m) \in \Omega_d := \left\{(h, m) \mid \sqrt{\frac{14}{15} d_*} \leq \|h\|_2 = \|m\|_2 \leq \sqrt{\frac{16}{15} d_*}\right\},
\end{equation}
for all $\pi(h, m) \in \Pi_\varepsilon$ and $\varepsilon \leq 1/15$.
There are unique orthogonal decompositions
\begin{equation*}
h = \tau_1 h_{\sharp} + \tilde{h} \hbox{ and } m = \tau_2 m_{\sharp} + \tilde{m},
\end{equation*}
where $\tilde{h}^* h_{\sharp} = 0$ and $\tilde{m}^* m_{\sharp} = 0$. Let
\begin{equation*}
\hat{h} = h - \tau h_{\sharp} \hbox{ and } \hat{m} = m - \bar{\tau}^{-1} m_{\sharp},
\end{equation*}
where $\tau = \frac{1}{(1 - \frac{\Delta}{10 d_*}) \bar{\tau}_2}$.

\begin{lemm} \label{BD:le11}
For all $(h, m)$ such that $(h, m) \in \Omega_d$ and $\pi(h, m) \in \Pi_\varepsilon$ with $\varepsilon \leq 1 / 15$, it holds that $\|\hat{h}\|_2^2 \leq 6.1 \Delta^2 / d_*$, $\|\hat{m}\|_2^2 \leq 6.1 \Delta^2 / d_*$, and $\|\hat{h}\|_2^2 \|\hat{m}\|_2^2 \leq 8.4 \Delta^4 / d_*^2$. Moreover, if we assume $\pi(h, m) \in \Omega_\mu$ additionally, we have $\sqrt{L} \|B (\hat{h})\|_\infty \leq 6 \mu \sqrt{d_*}$.
\end{lemm}
\begin{proof}
Inequalities $\|\hat{h}\|_2^2 \leq 6.1 \Delta^2 / d_*$, $\|\hat{m}\|_2^2 \leq 6.1 \Delta^2 / d_*$, and $\|\hat{h}\|_2^2 \|\hat{m}\|_2^2 \leq 8.4 \Delta^4 / d_*^2$ have been proven in~\cite[Lemma~5.15]{LLSW2016}. Note the definition of $\Omega_\mu$ and~\eqref{BD:e16}, we have $\{(h, m) \mid \sqrt{L} \|B h\|_{\infty} \|m\|_2 \leq 4 d_* \mu, \|h\|_2 = \|m\|_2\} \subset \{(h, m) \mid \sqrt{L} \|B h\|_{\infty} \leq 4.5 \sqrt{d_*} \mu\}$. We also have $|\tau_1 \bar{\tau_2}| \geq 1 - \varepsilon$ by \cite[Lemma 5.9]{LLSW2016} and $|\tau_1| \leq \frac{\|h\|_2}{\|h_{\sharp}\|_2} = \sqrt{ \frac{\|h\|_2\|m\|_2}{\|h_{\sharp}\|_2\|m_{\sharp}\|_2}} \leq \sqrt{\frac{\|h m^* - h_{\sharp} m_{\sharp}^*\|_F + \|h_{\sharp} m_{\sharp}^*\|_F}{d_*}} \leq \sqrt{1 + \varepsilon}$. Therefore, it holds that
\begin{align*}
\sqrt{L} \|B (\hat{h})\|_\infty \leq& \sqrt{L} \|B(h)\|_\infty + \frac{1}{(1 - \frac{\Delta}{10 d_*}) \bar{\tau}_2} \sqrt{L} \|B (h_{\sharp})\|_\infty \\
=& 4.5 \mu \sqrt{d_*} + \frac{\sqrt{1 + \varepsilon}}{(1 - \frac{\Delta}{10 d_*}) (1 - \varepsilon)} \leq 5.7 \mu \sqrt{d_*},
\end{align*}
where $\Delta / d_* \leq \varepsilon \leq 1 / 15$.
\end{proof}
Lemma~\ref{BD:e12} is the \cite[Lemm~5.16]{LLSW2016} and is used in the proof of Condition~\ref{BD:c3}.
\begin{lemm} \label{BD:le12}
For all $(h, m)$ such that $(h, m) \in \Omega_d$ and $\pi(h, m) \in \Pi_\varepsilon \cap \Omega_\mu$ with $\varepsilon \leq 1 / 15$, it uniformly holds that
\begin{equation*}
\mathrm{Re}\left(\inner[2]{\nabla_h^w f_{\mathrm{T}}}{\hat{h}} + \inner[2]{\nabla_m^w f_{\mathrm{T}}}{\hat{m}}\right) \geq \frac{\Delta^2}{8} - 2 \Delta \|\mathcal{A}^*(e)\|_2,
\end{equation*}
provided $L \geq a_3 \mu^2 (K + N) \log^2(L)$ for some constant $a_3$, where $\nabla^w$ denotes the Wirtinger derivative~\eqref{BD:e27}.
\end{lemm}

\begin{lemm} \label{BD:le13}
For all $(h, m)$ such that $(h, m) \in \Omega_d$ and $\pi(h, m) \in \Pi_\varepsilon$ with $\varepsilon \leq 1 / 15$ and $\frac{9}{10} d_* \leq d \leq \frac{11}{10} d_*$, it uniformly holds that
\begin{equation*}
\mathrm{Re}\left(\inner[2]{\nabla_h^w G_{\mathrm{T}}}{\hat{h}} + \inner[2]{\nabla_m^w G_{\mathrm{T}}}{\hat{m}}\right) \geq \frac{\Delta}{5 d_*} \sqrt{\rho G(\pi(h, m))},
\end{equation*}
where $\rho \geq d^2 + 2.5 \|e\|_2^2$.
\end{lemm}
\begin{proof}
We have
\begin{align*}
\inner[2]{\nabla_h G_{\mathrm{T}}}{\hat{h}} + \inner[2]{\nabla_m G_{\mathrm{T}}}{\hat{m}} =& \frac{L \rho}{4 d^2 \mu^2} \sum_{i = 1}^L G_0'\left(\frac{L|b_i^* h|^2 \|m\|_2^2}{8 d^2 \mu^2}\right) \inner[2]{b_i b_i^* h}{\hat{h}} \|m\|_2^2 \\
&+ \frac{L \rho}{4 d^2 \mu^2} \sum_{i = 1}^L G_0'\left(\frac{L|b_i^* h|^2 \|m\|_2^2}{8 d^2 \mu^2}\right) \inner[2]{m}{\hat{m}} |b_i^* h|^2
\end{align*}
When $L|b_i^* h|^2 \|m\|_2^2 \leq 8 d^2 \mu^2$, we have
\begin{align*}
\frac{L \rho}{4 d^2 \mu^2} \sum_{i = 1}^L G_0'\left(\frac{L|b_i^* h|^2 \|m\|_2^2}{8 d^2 \mu^2}\right) \inner[2]{b_i b_i^* h}{\hat{h}} \|m\|_2^2 =& 0 =\frac{\rho}{5} G_0'\left(\frac{L|b_i^* h|^2 \|m\|_2^2}{8 d^2 \mu^2}\right) \hbox{ and } \\
\frac{L \rho}{4 d^2 \mu^2} \sum_{i = 1}^L G_0'\left(\frac{L|b_i^* h|^2 \|m\|_2^2}{8 d^2 \mu^2}\right) \inner[2]{m}{\hat{m}} |b_i^* h|^2 =& 0 =\frac{\rho}{5} G_0'\left(\frac{L|b_i^* h|^2 \|m\|_2^2}{8 d^2 \mu^2}\right)
\end{align*}
When $L|b_i^* h|^2 \|m\|_2^2 > 8 d^2 \mu^2$, by \cite[Lemma~5.9]{LLSW2016} and~\eqref{BD:e16}, it holds that $|\tau_1 \bar{\tau}_2 - 1| \leq \frac{\Delta}{d_*}$, $|\tau_1| \leq 2$, which yields
\begin{equation*}
\tau = \frac{1}{(1 - \frac{\Delta}{10 d_*}) |\bar{\tau}_2|} = \frac{|\tau_1|}{(1 - \frac{\Delta}{10 d_*}) |\tau_1 \bar{\tau}_2|} \leq \frac{2}{(1 - \frac{\Delta}{10 d_*}) (1 - \frac{\Delta}{d_*})}.
\end{equation*}
By $\mu_h \leq \mu$ and $\Delta / d_* \leq \varepsilon \leq 1/15$, we obtain
\begin{align*}
\mathrm{Re}(\inner[2]{b_ib_i^*h}{\hat{h}}) =& \mathrm{Re}(|b_i^* h|^2 - \tau \inner[2]{b_i^* h}{b_i^* h_{\sharp}}) \\
\geq& |b_i^* h| (|b_i^* h| - \frac{2}{(1 - \frac{\Delta}{10 d_*}) (1 - \frac{\Delta}{d_*})} |b_i^* h_{\sharp}|) \\
\geq& |b_i^* h| (|b_i^* h| - \frac{2 \mu}{(1 - \frac{\Delta}{10 d_*}) (1 - \frac{\Delta}{d_*})} \sqrt{\frac{d_*}{L}}) \hbox{ (by $|b_i^* h_{\sharp}| \leq \mu_h \sqrt{d_* / L}$)} \\
\geq& \sqrt{\frac{8 d^2 \mu^2}{L \|m\|_2^2}} \left(\sqrt{\frac{8 d^2 \mu^2}{L \|m\|_2^2}} - \frac{2 \mu}{(1 - \frac{\Delta}{10 d_*}) (1 - \frac{\Delta}{d_*})} \sqrt{\frac{d_*}{L}}\right) \\
\geq& \sqrt{\frac{8 d^2 \mu^2}{L \|m\|_2^2}} \left(\sqrt{\frac{8 d^2 \mu^2}{L \|m\|_2^2}} - \frac{2 \mu}{(1 - \frac{\Delta}{10 d_*}) (1 - \frac{\Delta}{d_*})} \sqrt{\frac{10 d}{9 L} \frac{16}{15} \frac{10}{9} \frac{d}{\|m\|_2^2}}\right) \hbox{ (by~\eqref{BD:e16} and $d_* \leq \frac{10}{9}d$)} \\
\geq& \frac{4 d^2 \mu^2}{5 L \|m\|_2^2}.
\end{align*}
Therefore,
\begin{equation*}
\frac{L \rho \|m\|_2^2}{4 d^2 \mu^2} G_0'\left(\frac{L|b_i^* h|^2 \|m\|_2^2}{8 d^2 \mu^2}\right) \mathrm{Re}\left(\inner[2]{b_i b_i^* h}{\hat{h}}\right) \geq \frac{\rho}{5} G_0'\left(\frac{L|b_i^* h|^2 \|m\|_2^2}{8 d^2 \mu^2}\right).
\end{equation*}
We also have
\begin{align*}
\mathrm{Re}(\inner[2]{m}{\hat{m}}) = \mathrm{Re}(\inner[2]{m}{m - \bar{\tau}^{-1} m_{\sharp}}) = \frac{\Delta}{10 d_*} \|m\|_2^2 + \left(1 - \frac{\Delta}{10 d_*}\right) \|\tilde{m}\|_2^2 \geq \frac{\Delta}{10 d_*} \|m\|_2^2.
\end{align*}
It follows that
\begin{gather*}
\frac{L \rho}{4 d^2 \mu^2} G_0'\left(\frac{L|b_i^* h|^2 \|m\|_2^2}{8 d^2 \mu^2}\right) \mathrm{Re}(\inner[2]{m}{\hat{m}}) |b_i^* h|^2 \geq \frac{L \rho}{4 d^2 \mu^2} G_0'\left(\frac{L|b_i^* h|^2 \|m\|_2^2}{8 d^2 \mu^2}\right) \frac{\Delta}{10 d_*} \|m\|_2^2 \frac{8 d^2 \mu^2}{L \|m\|_2^2} \\
\geq \frac{\rho}{75} G_0'\left(\frac{L|b_i^* h|^2 \|m\|_2^2}{8 d^2 \mu^2}\right).
\end{gather*}
Therefore, we have
\begin{align*}
&\mathrm{Re}\left(\inner[2]{\nabla_h G_{\mathrm{T}}}{\hat{h}} + \inner[2]{\nabla_m G_{\mathrm{T}}}{\hat{m}}\right) \geq \sum_{i = 1}^L \frac{16\rho}{75} G_0'\left(\frac{L|b_i^* h|^2 \|m\|_2^2}{8 d^2 \mu^2}\right)  \\
=& \frac{32 \rho}{75} \sum_{i = 1}^L \sqrt{G_0\left(\frac{L|b_i^* h|^2 \|m\|_2^2}{8 d^2 \mu^2}\right)} \geq \frac{32 \rho}{75} \sqrt{\sum_{i = 1}^L G_0\left(\frac{L|b_i^* h|^2 \|m\|_2^2}{8 d^2 \mu^2}\right)} = \frac{32}{75} \sqrt{\rho G(\pi(h, m))}
\end{align*}
The final result follows from the above inequality, \eqref{BD:e27} and~$\Delta /d_* \leq 1/15$.
\end{proof}

\begin{proof}[\textbf{Proof of Condition~\ref{BD:c3}}]
Suppose $(h, m)$ satisfies $(h, m) \in \Omega_d$ and $\pi(h, m) \in \Pi_\varepsilon \cap \Omega_\mu$ with $\varepsilon \leq 1 / 15$.
Since Lemmas~\ref{BD:le11},~\ref{BD:le12} and~\ref{BD:le13} are exactly the same as \cite[Lemma~5.15, Lemma~5.16, and Lemma~5.17]{LLSW2016}, we have the same result as \cite[Lemma~5.18]{LLSW2016}, i.e.,
\begin{equation*}
\|\nabla^w \tilde{f}_{\mathrm{T}}(h, m)\|_2^2\geq \frac{d_*}{5000} [\tilde{f}(\pi(h, m)) - c]_+
\end{equation*}
By Lemma~\ref{BD:le14} and~\eqref{BD:e27}, $\|h\|_2 = \|m\|_2$ implies that
$$
\nabla^w \tilde{f}_{\mathrm{T}}(h, m) = \frac{1}{2} P_{(h, m)}^h\nabla \tilde{f}_{\mathrm{T}}(h, m) = \frac{\|h\|_2 \|m\|_2}{2} \left(\grad \tilde{f}(\pi(h, m))\right)_{\uparrow_{(h, m)}}.
$$
By definition of the induced-norm, we have
\begin{align}
\|\grad \tilde{f}(\pi(h, m))\|_{g_{\pi(h, m)}}^2 =& \frac{4}{\|h\|_2 \|m\|_2} \|\nabla^w \tilde{f}_{\mathrm{T}}(h, m)\|_2^2  \nonumber \\
\geq& \frac{1}{1500} [\tilde{f}(\pi(h, m)) - c]_+. \label{BD:e17}
\end{align}
Since $\|\grad \tilde{f}(\pi(h, m))\|_{g_{\pi(h, m)}}^2$ and $\tilde{f}(\pi(h, m))$ are independent of representation in $\pi^{-1}(\pi(h, m))$, we have that~\eqref{BD:e17} holds for all $\pi(h, m) \in \Pi_\varepsilon \cap \Omega_\mu$ with $\varepsilon \leq 1 / 15$.
\end{proof}

\begin{proof} [\textbf{Proof of Condition~\ref{BD:c4}}]
By~\cite[Condition~5.4]{LLSW2016}, we have
\begin{equation*}
\left\|\left(\nabla \tilde{f}_{\mathrm{T}}(t \eta_{\uparrow_{(h, m)}})\right) - \left(\nabla \tilde{f}_{\mathrm{T}}(0)\right)\right\|_{2} \leq \tilde{a}_L t \|\eta_{\uparrow_{(h, m)}}\|_{2}, \;\;\; \forall 0 \leq t \leq 1,
\end{equation*}
for some positive $\tilde{a}_L$. Choosing $(h, m)$ such that $\|h\|_2 = \|m\|_2$ yields that
\begin{equation*}
\frac{14}{15} d_* \|\xi_{\uparrow_{(h, m)}}\|_{2}^2 \leq \|\xi_{\uparrow_{(h, m)}}\|_{g}^2 \leq \frac{16}{15} d_* \|\xi_{\uparrow_{(h, m)}}\|_{2}^2,
\end{equation*}
and $\nabla \tilde{f}_{\mathrm{T}}(\xi_{\uparrow_{(h, m)}}) = \|h\|_2 \|m\|_2 \left(\grad \hat{f}_{\pi(h, m)}(\xi_{\pi(h, m)})\right)_{\uparrow_{(h, m)}}$ for all $\xi_{\pi(h, m)} \in \T_{\pi(h, m)} \mathcal{Q}_1$.
Therefore, we have
\begin{equation*}
\|\grad \hat{f}_{\pi(h, m)}(t \eta_{\pi(h, m)}) - \grad \hat{f}_{\pi(h, m)}(0)\|_{g_{(h, m)}} \leq \frac{8}{7} \frac{15}{14 d_*} \tilde{a}_L t \|\eta_{\pi(h, m)}\|_{g_{(h, m)}}, \;\;\; \forall \;\; 0 \leq t \leq 1,
\end{equation*}
where $\|h\|_2\|m\|_2 \geq 14 d_* / 15$ from~\eqref{BD:e16} is used.
Note that both sides are independent of representation of $(h, m) \in \pi^{-1}(\pi(h, m))$, we completed the proof.
\end{proof}


\end{document}